\documentclass[a4paper, onecolumn, accepted=2022-09-22]{quantumarticle}
\pdfoutput=1
\usepackage{amsmath,amsfonts,amsthm,amssymb, mathtools}
\usepackage{hyperref}
\renewcommand{\thesection}{\arabic{section}}
\newtheorem{lemma}{Lemma}[section]
\newtheorem{corollary}[lemma]{Corollary}
\newtheorem{theorem}[lemma]{Theorem}
\newtheorem{prop}[lemma]{Proposition}

\theoremstyle{definition} %Set Following proclamations in normal text.
\newtheorem{definition}[lemma]{Definition}

\theoremstyle{remark}

\newtheorem{remark}[lemma]{Remark}

\newtheorem{example}[lemma]{Example}

\newcommand{\dg}{\sp{\text{\rm o}}}

\begin{document}

\title{Spectral resolutions in  effect algebras}

    \author{Anna Jen{\v c}ov{\'a}}
    \email{jenca@mat.savba.sk}
    \author{Sylvia Pulmannov\'{a}}
    \email{pulmann@mat.savba.sk}
    \affiliation{Mathematical Institute, Slovak Academy of
Sciences, \v Stef\'anikova 49, SK-814 73 Bratislava, Slovakia}

%\date{}

\maketitle

\begin{abstract} Effect algebras were introduced as an abstract algebraic model for
Hilbert space effects representing quantum mechanical measurements.  We study additional
structures on an effect algebra $E$ that enable us to define  spectrality and spectral
resolutions for elements of $E$ akin to those of self-adjoint operators. These structures, called
compression bases, are special families of maps on $E$, analogous to the set of compressions on operator algebras,
order unit spaces or unital abelian groups. Elements of a compression base are in one-to-one correspondence with certain elements of $E$,
called projections. An effect algebra is called spectral if it has a distinguished
compression base with  two special properties:  the projection cover property (i.e., for every element $a$ in $E$  there is a smallest projection
majorizing $a$), and the so-called b-comparability property, which is an analogue of general comparability in operator algebras or unital abelian groups.
 It is shown that in a spectral archimedean effect algebra $E$, every $a\in E$ admits a unique rational spectral resolution and its properties are studied. If in addition $E$ possesses a separating set of states, then every element $a\in E$ is determined by its spectral resolution.
It is also proved that for some types of interval effect algebras (with RDP, archimedean divisible), spectrality of $E$
is
equivalent to spectrality of its universal group and the corresponding rational spectral resolutions are the same. In
particular, for convex archimedean effect algebras, spectral resolutions in $E$ are in
agreement with spectral resolutions in the corresponding order unit space.
\end{abstract}

\section{Introduction} \label{sc:Intro}

In the mathematical description of quantum theory, the yes-no measurements are represented by Hilbert space effects, that is, operators between 0 and $I$ on the Hilbert space
representing the given quantum system. More generally, any quantum measurement is described as a measure with values in the set of effects.
 The existence of spectral resolutions of self-adjoint operators on a Hilbert space plays an important role in quantum theory. It provides a connection between such operators and
  projection-valued measures, describing sharp measurements. Moreover, spectrality appears as a crucial property in operational derivations of quantum theory, see e.g.
  \cite{BH, chiribella, HArdy, wetering}.

Motivated by characterization of state spaces of operator algebras and JB-algebras, Alfsen and Shultz introduced a notion of spectrality for order unit spaces, see \cite{AS, AlSh}. A more algebraic approach to spectral order unit spaces was introduced in \cite{FPspectres}, see \cite{JenPul} for a comparison of the two approaches. The latter approach is based on the works by Foulis \cite{Fcomgroup, Fgc, Funig, Forc}, who studied a generalization of spectrality for partially ordered unital abelian groups.
In all these works,  the basic notion is that of a \emph{compression}, which generalizes the map $a\mapsto pap$ for a projection $p$ on
 a von Neumann algebra, or the projection onto the ideal generated by a sharp element in an interpolation group \cite{Good}.
One of the highlights of these works is the result that if there exists a suitable set of compressions with specified properties, each element has a
unique spectral resolution, or a \emph{rational} spectral resolution with values restricted to $\mathbb Q$ in the
case of partially ordered abelian groups, analogous to the spectral resolution of self-adjoint elements in von
Neumann algebras.

As an algebraic abstraction of the set of Hilbert space effects, effect algebras were introduced by Foulis and Bennett \cite{FoBe}.
Besides the Hilbert space effects, this abstract definition covers a large class of other structures with no clear notion of a spectral resolution.
It is therefore important to study additional structures on an effect algebra that may ensure the existence of some type of a spectral resolution.
In analogy with the definition in \cite{Fcompog}, compressions and compression bases in effect algebras were studied by Gudder
\cite{Gu, Gucomprba}. It is a natural question (remarked upon also in \cite{Gu}) whether
some form of spectrality can be obtained in this setting and what structural properties of the
effect algebra are required for existence of spectral resolutions. Note that another
approach to spectrality in \emph{convex} effect algebras, based on contexts, was studied
in \cite{Gucs}, see also \cite{JePl}. There is also a dual notion of spectrality of
compact convex sets pertaining the state space of the effect algebra (e.g. \cite{AlSh,BH}),
this will not be considered here.

Spectral compression bases in effect algebras were studied in \cite{Pucompr}, building on the works \cite{Gucomprba,Forc}.
For an effect algebra $E$, a compression is an additive idempotent mapping $J: E\to E[0,p]$ where $p$ is a
special element called the focus of $J$. A compression base is a family $(J_p)_{p\in P}$
of compressions satisfying certain properties, parametrized by a specified set $P$ of
focuses. Elements of $P$ are called projections. We say that a compression base $(J_p)_{p\in P}$ in $E$ is spectral if it has (1) the projection cover property, that is,
for every $a\in E$ there is a smallest projection that majorizes $a$; (2) the so-called b-comparability, introduced in
analogy with general comparability in groups of \cite{Fgc} and \cite{Good}.

In the present paper, we study the properties of effect algebras with spectral compression bases.  We show that
 in such a case, every element $a\in E$ admits a unique spectral resolution in terms of
 spectral projections in $P$, parametrized by dyadic rationals. If $E$ is also  archimedean, the spectral resolution  can be extended to all rational values in $[0,1]$ and
characterized by several properties resembling the characterization of the spectral resolution for Hilbert space effects.
  We prove that for any state $s$ on $E$,  $s(a)$ is determined by the
values of $s$ on the spectral projections. In particular, if $E$ has a separating set of states, then every element is
uniquely determined by its spectral resolution. 

We also study some special cases of interval effect algebras, that is, effect algebras that are isomorphic to the unit
interval in a unital abelian partially ordered group $(G,u)$: effect algebras with RDP, divisible and convex archimedean
effect algebras. We show that in these cases, the compression bases in $E$ are in one-to-one correspondence with
compression bases in $G$ and that spectrality in $E$ is equivalent to spectrality in $G$, as defined in \cite{Forc}. Moreover, the spectral
resolutions obtained in $E$ coincide with the spectral resolutions in $G$. However, we
show that spectrality of an interval effect algebra is not always inherited from its
universal group.

\section{Effect algebras}

An \emph{effect algebra} \cite{FoBe} is a system $(E; \oplus,0,1)$ where $E$ is a nonempty set, $\oplus$ is a partially defined binary operation on $E$, and $0$ and $1$ are constants, such that the following conditions are satisfied:

\begin{enumerate}
\item[(E1)] If $a\oplus b$ is defined then  $b\oplus a$ is defined and $a\oplus b=b\oplus a$.
\item[(E2)] If $a\oplus b$ and $(a\oplus b)\oplus c$ are defined then  $b\oplus c$ and $a\oplus(b\oplus c)$ are defined and
$(a\oplus b)\oplus c=a\oplus(b\oplus c)$.
\item[(E3)] For every $a\in E$ there is a unique $a'\in E$ such that $a\oplus a'=1$.
\item[(E4)]If  $a\oplus 1$ is defined then $a=0$.
\end{enumerate}

\begin{example}\label{ex:effects} Let $\mathcal H$ be a Hilbert space and let $E(\mathcal H)$ be the set of operators on $\mathcal H$ such
 that $0\le A\le I$. For  $A,B\in E(\mathcal H)$, put $A\oplus B=A+B$ if $A+B\le I$,
otherwise $A\oplus B$ is not defined. Then $(E(\mathcal H);\oplus, 0, I)$ is an effect algebra. This is a prototypical
example on which the above abstract definition is modelled. The elements of $E(\mathcal H)$ are called effects.

\end{example}

\begin{example}\label{ex:intervalea} Let $(G,u)$ be a partially ordered abelian group with an order unit $u$. Let $G[0,u]$ be the unit interval in $G$ (we will often write $[0,u]$ if the group $G$ is clear). For $a,b\in G[0,u]$, let  $a\oplus b$ be defined if $a+b\le u$ and in this case $a\oplus b=a+b$. It is easily checked that $(G[0,u],\oplus, 0,u)$ is an effect algebra. Effect algebras of this form are called \emph{interval} effect algebras. In particular, the real unit interval $\mathbb R[0,1]$ can be given a structure of an effect algebra. Note also that the Hilbert space effects in Example \ref{ex:effects} form an interval effect algebra.

\end{example}

We write $a\perp b$ and say that $a$ and $b$ are \emph{orthogonal} if $a\oplus b$ exists. In what follows, when we write $a\oplus b$, we tacitly assume that $a\perp b$. A partial order is introduced on $E$ by defining $a\leq b$ if there is $c\in E$ with $a\oplus c=b$.  If such an element $c$ exists, it is unique, and we define $b\ominus a:=c$.  With respect to this partial order we have $0\leq a\leq 1$ for all $a\in E$. The element $a'=1\ominus a$ in (E3)  is called the \emph{orthosupplement} of $a$. It can be shown that $a\perp b$ iff $a\leq b'$ (equivalently, $b\leq a'$). Moreover $a\leq b$ iff $b'\leq a'$, and $a''=a$.

An element $a\in E$ is called \emph{sharp} if $a\wedge a'=0$ (i.e., $x\leq a,a' \implies x=0$). We denote the set of all sharp elements of $E$ by $E_S$. An element $a\in E$ is \emph{principal} if $x,y\leq a$, and $x\perp y$ implies that $x\oplus y\leq a$. It is easy to see that a principal element is sharp.

By recurrence, the operation $\oplus$ can be extended to finite sums $a_1\oplus a_2\oplus \cdots \oplus a_n$ of (not necessarily different) elements $a_1,a_2,\ldots a_n$ of $E$.
If $a_1=\dots=a_n=a$ and $\oplus_i a_i$ exist, we write $\oplus_i a_i=na$. An effect algebra $E$ is \emph{archimedean} if for $a\in E$, $na\le 1$ for all $n\in \mathbb N$ implies that $a=0$.

An infinite family $(a_i)_{i\in I}$ of elements of $E$ is called \emph{orthogonal} if every its finite subfamily has an $\oplus$-sum  in $E$.
If the element  $\oplus_{i\in I} a_i=\bigvee_{F\subseteq I}\oplus_{i\in F}a_i$, where the supremum is taken over all finite subsets of $I$ exists, it is called the \emph{orthosum} of
the family $(a_i)_{i\in I}$. An effect algebra $E$ is called \emph{$\sigma$-orthocomplete} if the orthosum exists for any  orthogonal countable subfamily of $E$.

An effect algebra $E$ is \emph{monotone $\sigma$-complete} if every ascending sequence $(a_i)_{i\in {\mathbb N}}$  has a supremum $a=\bigvee_i a_i$
   in $E$, equivalently, every descending sequence $(b_i)_{i\in {\mathbb N}}$ has an infimum $b=\bigwedge_i b_i$  in $E$. It turns out that an effect algebra is monotone
   $\sigma$-complete if and only if it is $\sigma$-orthocomplete \cite{JePu}. A subset $F$ of $E$ is \emph{sup/inf -closed in E} if whenever  $M\subseteq F$ and $\wedge M$ ($\vee M$) exists in $E$, then $\wedge M \in F$ ($\vee M \in F$).

If $E$ and $F$ are effect algebras, a mapping $\phi: E\to F$ is \emph{additive} if $a\perp b$ implies $\phi(a) \perp \phi(b)$ and $\phi(a\oplus b)=\phi(a)\oplus \phi(b)$. An additive mapping $\phi$ such that $\phi(1)=1$, is a \emph{morphism}. If $\phi :E\to F$ is a morphism, and $\phi(a)\perp \phi(b)$ implies $a\perp b$, then $\Phi$ is a \emph{monomorphism}. A surjective monomorphism is an \emph{isomorphism}.

A \emph{state} on an effect algebra $E$ is a morphism $s$  from $E$ into the effect algebra ${\mathbb R}[0,1]$ . We denote the set of states on $E$ by $S(E)$. We say that
$S\subset S(E)$ is \emph{separating} if $s(a)= s(b)$ for every $s\in S$ implies that $a= b$. We say that $S\subset S(E)$ is \emph{ordering} (or order determining) if
$s(a)\leq s(b)$ for all $s\in S$ implies $a\leq b$. If $S$ is ordering, then it is separating, the converse does not hold.

A lattice ordered effect algebra $M$, in which  $(a\vee b)\ominus a=a\ominus (a\wedge b)$ holds for all $a,b\in M$, is called an \emph{MV-effect algebra}.
We recall that MV-effect algebras are equivalent with MV-algebras, which were introduced by \cite{Chang} as algebraic bases for many-valued logic. It was proved in \cite{Mun}
that there is an equivalence between  MV-algebras and   lattice ordered groups with order unit, in the sense of category theory.

\section{Compressions on  effect algebras }

The definitions in this section  follow the works of Foulis {\cite{Fcomgroup,Fcompog,Fgc, Forc}},  who studied
spectrality for unital partially ordered abelian groups. Some details on the constructions
in that setting are given in Appendix \ref{app:spec}. The definitions were adapted to effect
algebras by  Gudder \cite{Gu} and Pulmannov\'a \cite{Pucompr}.

\begin{definition}\label{de:compr} Let $E$ be an effect algebra.
\begin{enumerate}
\item An additive map $J: E\to E$ is a
\emph{retraction} if $a\leq J(1)$ implies $J(a)=a$.
\item The element $p:=J(1)$ is called the \emph{focus} of $J$.
\item A retraction is a \emph{compression} if $J(a)=0\ \Leftrightarrow\ a\leq p'$.
\item If $I$ and $J$ are retractions we say that $I$ is a \emph{supplement of} $J$ if $\ker(J)=I(E)$ and $\ker(I)=J(E)$.
\end{enumerate}
\end{definition}

It is easily seen that any retraction is idempotent, in fact, an additive idempotent map is a retraction if and only if its range is an interval. If a retraction $J$ has a supplement $I$, then both $I$ and $J$ are
compressions and $I(1)=J(1)'$. The element $p:=J(1)$ for a retraction $J$ is called the focus of $J$. The focus is a principal element and we have $J(E)=[0,p]$,
moreover, $J$ is a compression if and only if $\mathrm{Ker}(J)=[0,p']$. For these and further properties see \cite{Gucomprba, Pucompr}.

\begin{example}\label{ex:compr_effects} Let $E(\mathcal H)$ be the algebra of effects on $\mathcal H$ (Example
\ref{ex:effects}) and let $p\in E(\mathcal H)$ be a projection. Let us define the map $J_p:a\mapsto pap$, then $J_p$ is a compression on
$E(\mathcal H)$ and $J_{p'}$ is a supplement of $J_p$. By \cite{Fcompog}, any retraction on $E(\mathcal H)$
is of this form for some projection $p$. In particular, any projection is the focus of a unique retraction $J_p$ with a
(unique) supplement $J_{p'}$. Effect algebras such that any retraction is supplemented  and  uniquely determined by its
focus are called compressible, \cite{Gu,Fcomgroup}.

\end{example}

Our aim is to introduce a set of compressions analogous to $(J_p)_{p\in P(\mathcal H)}$ in the above example  for any effect algebra $E$. For this, we need to specify a set $P\subseteq E$ of projections and the corresponding compressions. We will start with  the assumption that $P\subseteq E$ is a sub-effect algebra.

\begin{lemma} Let $P\subseteq E$ be a sub-effect algebra and let $(J_p)_{p\in P}$ be a set of compressions such that $J_p(1)=p$. Then $P$ is an orthomodular poset and
 $J_{p'}$ is a supplement of $J_p$.

\end{lemma}

\begin{proof} We refer the reader to \cite[\S 1.5]{DvPu} for the definition and characterization of orthomodular posets (OMP). Since the elements of  $P$  are principal hence sharp, $P$ is an orthoalgebra. To show that $P$ is an OMP, it is enough to show that if $p\oplus q$ exists, then
$p\oplus q=p\vee q$. So assume that $p\oplus q$ exists, then $p,q\le p\oplus q$. Suppose that $r\in P$ is such that $p,q\le r$, then since $r$ is principal, we have $p\oplus q\le r$. It follows that $p\oplus q=p\vee q$ and $P$ is an OMP. Further, since any $J_p$ is a compression, for $p\in P$ we have $J_{p'}(E)= [0,p']=\mathrm{Ker}(J_p)$, so that $J_{p'}$ is a supplement of $J_p$.

\end{proof}

The set of compressions in $E(\mathcal H)$ has another important property: $J_p$ and $J_q$ commute if and only if the projections $p$ and $q$ commute, and in this case, $J_p\circ J_q=J_q\circ J_p=J_{pq}$.
To extend this property to general effect algebras, we need the following notion.
Recall that two elements $a,b$ in an effect algebra $E$ are \emph{Mackey compatible} if there are elements $a_1,b_1, c\in E$ such that $a_1\oplus b_1\oplus c$ exists and $a=a_1\oplus c, b=b_1\oplus c$.  We shall write $a\leftrightarrow b$ if
$a,b$ are Mackey compatible. It is well known that for two projections (or a projection
and an effect) in $E(\mathcal H)$, Mackey compatibility is equivalent to commutativity.
(Note that this is not true for general pairs of elements.)

We can now introduce the main definition of this section.

\begin{definition}\label{def:compressionbase} A family $(J_p)_{p\in P}$ of compressions on an effect algebra $E$, indexed by a sub-effect algebra $P\subseteq E$
is called a  \emph{compression base} on  $E$ if the following conditions hold:
\begin{enumerate}
\item[(C1)] each $p\in P$ is the focus of $J_p$,

\item[(C2)] if $p,q\in P$ are Mackey compatible (in $E$), then $J_p\circ J_q=J_r$ for some $r\in P$.

\end{enumerate}
Elements of $P$ are  called \emph{projections}.

\end{definition}

The above definition is slightly different from the original definitions of a compression base used in \cite{Gucomprba}. We next show that these definitions are equivalent.
For this, recall that a sub-effect algebra $F\subseteq E$ is  \emph{normal} if  for all
$e,f,d\in E$ such that  $e\oplus f\oplus d$ exists in $E$, we have $e\oplus d, f\oplus
d\in F\ \implies\ d\in F$. It was proved in \cite{Gucomprba} that in this case, two
elements in $F$ are Mackey compatible as elements of $E$ if and only if  they are Mackey
compatible in $F$, that is, the elements $a_1,b_1$ and $c$ can be chosen in $F$.

\begin{prop}\label{prop:CB_orig} Let $(J_p)_{p\in P}$ be a compression base. Then
\begin{enumerate}
\item [(i)] $P$ is a normal sub-effect algebra in $E$;
\item[(ii)]  if $p,q,r\in P$ are such that $p\oplus q\oplus r$ exists, then $J_{p\oplus q}\circ J_{q\oplus r}=J_r$.
 \end{enumerate}
 Conversely, if $(J_p)_{p\in P}$ is a set of compressions such that $J_p(1)=p$, and  with properties (i) and (ii), then it is a compression base.

\end{prop}

\begin{proof} Let $(J_p)_{p\in P}$ be a compression base and let $e,f,d\in E$ be such that $e\oplus f\oplus d\in E$ and $p=e\oplus d, q=f\oplus d\in P$. Then $p\leftrightarrow q$ (in $E$), so that $J_p\circ J_q=J_r$ for some $r\in P$. We have
\[
r=J_r(1)=J_p\circ J_q(1)=J_p(q)=J_p(f)\oplus J_p(d)=d,
\]
since $f\le p'$ and $d\le p$. It follows that $d=r\in P$, so that $P$ is normal. This also proves (ii). The converse is easy.

\end{proof}

\begin{corollary} Let $(J_p)_{p\in P}$ be a compression base. Then $p\leftrightarrow q$ if and only if $J_p\circ J_q=J_q\circ J_p=J_{p\wedge q}$.

\end{corollary}

\begin{proof}
 Assume that $p\leftrightarrow q$ and let $p=p_1\oplus r$ and $q=q_1\oplus r$ be such that
 $p_1\oplus q_1\oplus r$ exists. Then $p\wedge p=r$ and we have by the property (ii) in
 Proposition \ref{prop:CB_orig} that $J_p\circ J_q=J_r=J_q\circ J_p$. The converse follows
 by \cite[Thm. 4.2]{Gucomprba}: assume that $J_p\circ J_q=J_r$ for some $r\in P$, then
 $J_p(q)=J_r(1)=r\le p$ and $J_r(q)=J_p\circ J_q(q)=J_p(q)=r$, so that $r\le q$. Put $q_1=q\ominus r$, $p_1=p\ominus r$, then $J_p(q_1)=r\ominus r=0$, so that $q_1\le p'$. It follows that $p_1\oplus r\oplus q_1$ exists and $p\leftrightarrow q$.

\end{proof}

By  \cite[Cor. 4.5]{Gu} and \cite[Thm. 2.1]{Pucompr},  the set $P$ as a subalgebra of $E$ is a regular
orthomodular poset (OMP) with the orthocomplementation $p\mapsto p'$. Recall
that an OMP $P$ is \emph{regular} if for all $a,b,c\in P$, if $a,b$ and $c$ are pairwise
Mackey compatible, then
$a\leftrightarrow b\vee c$ and $a\leftrightarrow b\wedge c$, \cite{PtPu, Harding}.

\begin{example} Let $P(\mathcal H)$ be the set of all projections on a Hilbert space $\mathcal H$ and let $J_p$ for
$p\in P(\mathcal H)$ be as
in Example \ref{ex:compr_effects}. It is easily observed that the set $(J_p)_{p\in P(\mathcal H)}$ is a compression base
in $E(\mathcal H)$. More generally, the set of all compressions in a compressible effect algebra is a compression base,
\cite{Gucomprba}.

\end{example}

In general, there can be many compression bases on an effect algebra $E$, for example,
$(J_0=0,J_1=id)$ is trivially a compression base. The compression bases on $E$ can be ordered by inclusion and by standard Zorn Lemma arguments, there are maximal elements with respect to this order, called maximal compression bases, \cite{Gucomprba}.
The following example is a compression base on any effect algebra $E$.

\begin{example}\label{ex:center} Let $E$ be an effect algebra and let $\Gamma(E)$ be the
set of \emph{central elements} \cite[Def. 1.9.11, Lemma 1.9.12]{DvPu}. In this case, both 
 $p$ and  its orthosupplement $p'$ are  principal elements, $a\wedge p$, $a\wedge p'$
 exist  and we have
\[
a=a\wedge p\oplus a\wedge p'
\]
for all $a\in E$.  It is easily seen that $U_p: a\mapsto a\wedge p$ is the unique compression with focus $p\in \Gamma(E)$ and $(U_p)_{p\in \Gamma(E)}$ is a compression base. This compression base will be called \emph{central}.

\end{example}
The next example shows that there may exist different compression bases with the same  set of projections $P$.

\begin{example}\label{ex:horsum} Let $E=E(\mathcal H)\dot{\cup}E(\mathcal H)$ be the
horizontal sum (or the 0-1 pasting) of two copies of $E(\mathcal H)$ for a separable
Hilbert space $\mathcal H$. That is, the elements of $E$ are of the form $(a,0)$ or
$(0,a)$ where $a\in E(\mathcal H)$ and $(1,0)$ and $(0,1)$ are identified with the unit
element $u$ of $E$.
Two elements of $E$ are summable if and only if they are of the form $(a,0)$, $(b,0)$ or
$(0,a)$, $(0,b)$ for some summable $a,b\in E(\mathcal H)$ and then $(a,0)\oplus
(b,0)=(a\oplus b,0)$, similarly in the other copy. The sharp elements of $E$ are of the form $(p,0)$ or $(0,p)$ for some $p\in P(\mathcal H)$ and all such elements are principal.

Choose any faithful state $\varphi$ on $E(\mathcal H)$, that is, $\varphi(a)=0$ for $a\in
E(\mathcal H)$ implies $a=0$. For any $p\in P(\mathcal H)$, $p\ne 1$,  we define
\[
J_{(p,0)}((a,0))=(J_p(a),0),\quad J_{(p,0)}((0,a))=(\varphi(a)p,0)
\]
It is easily checked that $J_{(p,0)}$ is a compression on $E$ with focus $(p,0)$. The
compressions $J_{(0,p)}$ can be defined similarly. We also put  $J_u=id$.

Let $P=P(\mathcal H) \dot{\cup} P(\mathcal H)$, then $P$ is a subalgebra in $E$ and we can
check that $(J_e)_{e\in P}$ is a compression base. Indeed, let $e,f\in P$ and
$e\leftrightarrow f$, we may clearly assume that none of the projections is the unit
element. Then $e,f$ must belong to the same copy, hence there are some commuting
projections $p,q\in P(\mathcal H)\setminus \{1\}$, such that $e=(p,0)$ and $f=(q,0)$, say. Then
\begin{align*}
J_e\circ J_f((a,0))=(J_p\circ J_q(a),0)=(J_{pq}(a),0),\qquad J_e\circ J_f((0,a))=(J_p(\varphi(a)q),0)=(\varphi(a)pq,0),
\end{align*}
so that $J_e\circ J_f=J_g$ with $g=(pq,0)$. Notice that the state $\varphi$ can be chosen
arbitrarily, this shows that there can be different compression bases with the same set of
projections $P$. This example will be continued in Example \ref{ex:horsum_spectral}.

More generally, if $E$ is the horizontal sum of two effect algebras $E_1$ and $E_2$, then all principal elements are of the form $(p_1,0)$ or $(0,p_2)$ for some principal $p_1\in E_1$ or $p_2\in E_2$. If $p_1\ne 0,1$,  then  any retraction $J$ on $E$ with focus $(p_1,0)$ has the form
\[
J((a_1,0))=(J_1(a),0),\qquad J((0,a_2))=(\phi(a_2),0),
\]
where $J_{1}$ is a retraction on $E_1$ with focus $p_1$ and $\phi:E_2\to [0,p_1]$ is a
morphism, see \cite{Gu}. Note that $J$ is a compression if and only if $J_1$ is a compression and $\phi$ is faithful, that is, $\phi(a_2)=0$ implies $a_2=0$. Even if both $E_1$ and $E_2$ have compressions,  faithful morphisms $\phi$ might not exist, so in that case there are no nontrivial compressions on $E$ (take for example $E_1=E_2$ the Boolean algebra $2^2$).

\end{example}

\begin{example}\label{ex:seq} An effect algebra $E$ is \emph{sequential} if it is endowed with a binary operation
$\circ$ called a sequential product, see \cite{GuGr} for definition and more information. In this case, an element $p\in
E$ is sharp if and only if $p\circ p=p$, equivalently, $p\circ p'=0$. Moreover, for $a\in E$, $a\le p$ if and only if
 $p\circ a = a$, equivalently, $p'\circ a=0$. From the axioms of the sequential product and
these properties it immediately follows that $J_p:a\mapsto p\circ a$ is a compression and $(J_p)_{p\in E_S}$ is a
compression base, see \cite{Gucomprba}. Note that $E(\mathcal H)$ is sequential, with $a\circ b=a^{1/2}ba^{1/2}$.

\end{example}

\begin{example}\label{ex:omp} Let $E$ be an OMP. Notice that all elements are principal, but in general there is no compression base
such that all elements are projections. Indeed, by \cite[Thm. 4.2]{Gucomprba}, for projections $p$ and $q$, $J_p(q)$ is
a projection if and only if $p\leftrightarrow q$. It follows that there is a compression base $(J_p)_{p\in E}$ if and
only if all elements are compatible, that is $E$ is a Boolean algebra. In that case we have $J_p(q)=U_p(q)=p\wedge q$.

\end{example}

\begin{example}\label{example:rdp} We say that an effect algebra $E$ has the Riesz decomposition property (RDP) if one of the following
equivalent conditions is satisfied:
\begin{enumerate}
\item[(i)] For any $a,b,c\in E$, if $a\le b\oplus c$ then there are some $b_1,c_1\in E$ such that $b_1\le b$, $c_1\le c$
and $a=b_1\oplus c_1$.
\item[(ii)] For any $a_1,a_2,b_1,b_2\in E$, if  $a_1\oplus a_2=b_1\oplus b_2$, then there are elements $w_{ij}\in E$,
 $i,j=1,2$ such that $a_1=w_{11}\oplus w_{12}$, $a_2=w_{21}\oplus w_{22}$, $b_1=w_{11}\oplus w_{21}$ and
$b_2=w_{12}\oplus w_{22}$.

\end{enumerate}
By \cite{Rav}, $E$ is an interval effect algebra. If, in addition, $E$ is lattice-ordered, then $E$ is an MV-effect algebra, see \cite{DvPu}.

From the condition (i) (or (ii)), it follows that all elements in $E$ are pairwise Mackey compatible. Moreover, every sharp element in $E$
is principal and the set $E_S$ of sharp elements is the center of $E$. By Example \ref{ex:center},  $(U_p)_{p\in E_S}$
is a compression base.  Since the focus of a retraction is always sharp, this compression base contains all retractions on $E$. In this case, we say that the compression base is \emph{total}.

\end{example}

\textbf{From now on, we will assume that the effect algebra $E$ is endowed with a fixed
compression base $(J_p)_{p\in P}$ and all further considerations are made with respect
to this base.}

\subsection{Compatibility and commutants}
\label{sec:commutants}

By \cite[Lemma 4.1]{Pucompr} we have the following.

\begin{lemma}\label{le:comE} If $p\in P, a\in E$, then the following statements are equivalent:
\begin{enumerate}
\item $J_p(a)\leq a$,
\item $a=J_p(a)\oplus J_{p'}(a)$,
\item $a\in E[0,p]\oplus E[0,p']$,
\item $a \leftrightarrow p$,
\item $J_p(a)=p\wedge a$.
\end{enumerate}
\end{lemma}

If any of the conditions in the above lemma is satisfied for  $a\in E, p\in P$, we say that $a$ and $p$ \emph{commute}
or that $a$ and $p$ are \emph{compatible}. The \emph{commutant} of $p$ in $E$ is defined by
\[
C(p):=\{ a\in E: a=J_p(a) \oplus J_{p'}(a)\}.
\]
If $Q\subseteq P$, we write $C(Q):=\bigcap_{p\in Q}C(p)$. Similarly, for an element $a\in E$, and a subset $A\subseteq E$, we write
\[
PC(a):=\{p\in P,\ a\in C(p)\},\ \  PC(A):=\bigcap_{a\in A} PC(a).
\]
We also define
\begin{align*}
CPC(a):= C(PC(a)),\qquad 
P(a):=CPC(a)\cap PC(a)=PC(PC(a)\cup\{a\}).
\end{align*}
The set $P(a)\subseteq P$ will be called the \emph{bicommutant} of $a$. For a subset $Q\subseteq E$, we put
\[
P(Q):= PC(PC(Q)\cup Q).
\]
Note that the elements in $P(Q)$ are pairwise Mackey
compatible and since $P$ is a regular OMP, this implies that $P(Q)$ is a Boolean subalgebra in $P$.

\begin{lemma} \cite[Lemma 4.2]{Pucompr}\label{lemma:compatible_projs} Let $p,q\in P$, $a\in E$.
 If $p\perp q$ and either $a\in C(p)$ or $a\in C(q)$, then
\[
J_{p\vee q}(a)=J_{p\oplus q}(a)=J_p(a)\oplus J_q(a).
\]

\end{lemma}

Recall that a maximal set of pairwise Mackey compatible elements in a regular  OMP $P$ is
called a \emph{block} of $P$ \cite[Cor. 1.3.2]{PtPu}.
It is well known that every block  $B$ is a Boolean subalgebra of $P$ \cite[Thm. 1.3.29]{PtPu}.
If $B$ is a block of $P$, the set $C(B)$ will be called a \emph{C-block} of $E$.

\begin{example} \label{ex:hilb_commutant}
 Let  $E=E(\mathcal H)$ for a Hilbert space $\mathcal H$. It is easily checked that for any projection $p\in P(\mathcal H)$,
\[
C(p)=\{p\}'\cap E=\{a\in E, pa=ap\}
\]
and for any $a\in E$,
\[
P(a)=\{a\}''\cap P(\mathcal H)
\]
(here $C'$ denotes the usual commutant of a subset of bounded operators $C\subset B(\mathcal H)$). The C-blocks are the unit intervals in maximal abelian von Neumann subalgebras of $B(\mathcal H)$.
 In Section \ref{sec:bcomp} we introduce a property under which the C-blocks can be characterized in a similar way (see Theorem \ref{thm:cblock}).
\end{example}

\subsection{Projection cover property}

\begin{definition}\label{de:projcov} If $a\in E$ and $p\in P$, then $p$ is a \emph{projection cover} for $a$ if, for all
$q\in P$, $a\leq q \,\Leftrightarrow \, p\leq q$. We say that $E$  has the \emph{projection cover property} if every
effect $a\in E$ has a (necessarily unique) projection cover. The projection cover of $a\in
E$ will be denoted as $a\dg$.

\end{definition}

%If a compression base $(J_p)_{p\in P}$ has the projection cover property, we will sometimes say, in short, that $E$
 %has the projection cover property.

\begin{theorem}[{\cite[Thm. 5.2]{Gucomprba}, \cite[Thm. 5.1]{Pucompr}}]\label{th:projcovoml}  Suppose that $E$ has the projection cover property.
Then $P$ is an orthomodular
lattice (OML). Moreover, $\mathcal P$ is sup/inf-closed in $E$.
\end{theorem}

\begin{prop}\label{prop:pc_commutant} Let $E$ have the projection cover property. Then for
any  $a\in E$, $a\dg\in P(a)$.

\end{prop}

\begin{proof}  Since $a\le a\dg$, $a\in C(a\dg)$ by Lemma \ref{le:comE} (iv). The rest follows by \cite[Thm. 5.2 (i)]{Pucompr}.

\end{proof}

\subsection{b-property and comparability}\label{sec:bcomp}

In this section we recall the notion of the b-comparability property in effect algebras, introduced in \cite{Pucompr} as
an analogue of the general comparability property in unital partially ordered abelian groups
\cite{Funig}, which is itself an extension of the general comparability property in interpolation groups \cite{Good}. Some more
details on compressions and comparability for partially ordered abelian groups can be
found in Appendix \ref{app:spec}.

We first introduce a property that allows us to extend the notion of commutativity to all pairs of elements $a,b\in E$, obtaining the analogue of commutativity of quantum effects.

\begin{definition}[{\cite[Definition 6.1]{Pucompr}}]\label{de:belement}  We will say that $a\in E$ has the
\emph{b-property} (or is a \emph{b-element}) if there is a Boolean subalgebra $B(a)\subseteq P$ such that for all $p\in P$, $a\in C(p) \,\Leftrightarrow \, B(a)\subseteq C(p)$. We say that $E$ has the \emph{b-property} if every $a\in E$ is a b-element.
\end{definition}

The Boolean subalgebra  $B(a)$ in the above definition is in general not unique. The next lemma shows that the bicommutant of $a$ is the
 largest such subalgebra.

\begin{lemma}\label{lemma:BaPa} Let $a\in E$ be a b-element. Then
\begin{enumerate}
\item $B(a)\subseteq P(a)$.
\item For $p\in P$, $a\in C(p)$ $\iff$ $P(a)\subseteq C(p)$.
\end{enumerate}

\end{lemma}

\begin{proof} Let $q\in B(a)$. Since $B(a)$ is a Boolean subalgebra, all elements are mutually compatible, which means
that $B(a)\subseteq C(q)$. By definition, this implies $a\in C(q)$. Further, for any $p\in PC(a)$ we have $q\in B(a)\subseteq
C(p)$, so that
\[
q\in PC(a)\cap CPC(a)=P(a).
\]
This proves (i). To prove (ii), let $p\in P$ be such that $a\in C(p)$, then $p\in PC(a)$ so that $P(a)\subseteq C(p)$.
The converse follows by (i).

\end{proof}

\begin{lemma}\label{le:belement}{\rm \cite[Prop. 6.1]{Pucompr}}  {\rm(i)} If an element $a\in E$ is a b-element,
then there is a block $B$ of $P$ such that $a\in C(B)$. {\rm(ii)} Every projection $q\in P$ is a b-element with
$B(q)=\{0,q,q',1\}$.
\end{lemma}

Note that (ii) of the above Lemma shows that we may have $B(a)\subsetneq P(a)$. Indeed, if $E$ has RDP and $q\in E_S$,
 then $P(q)$ is all of $E_S$, see Example \ref{example:rdp}, whereas $B(q)$ in (ii) is the minimal Boolean subalgebra
with the required properties.

For $A,B\subseteq E$ we write $A\leftrightarrow B$ if $a\leftrightarrow b$ for all $a\in A$, $b\in B$.

\begin{lemma}\label{lemma:bcommuting} Let $e,f\in E$ be b-elements. Then $B(e)\leftrightarrow B(f)$ if and only if
$P(e)\leftrightarrow P(f)$.

\end{lemma}

\begin{proof} By definition of $B(e)$, $B(f)$ and Lemma \ref{lemma:BaPa}, we have the following chain of equivalences:
\begin{align*}
B(e)\leftrightarrow B(f) &\iff\ B(e)\subseteq C(B(f))  \iff \ e\in
C(B(f))\\ &\iff \ P(e)\subseteq C(B(f))
\iff \ P(e)\leftrightarrow B(f),
\end{align*}
the proof is finished by symmetry in $e$ and $f$.

\end{proof}

Let $e,f\in E$ have the b-property. We say that $e$ and $f$ \emph{commute}, in notation $eCf$, if

\begin{equation}\label{eq:commut}\ P(e) \leftrightarrow P(f).
\end{equation}
Clearly, for $p\in P$ we have $eCp \iff e\leftrightarrow p \iff e\in C(p)$,   \cite[Lemma 6.1]{Pucompr}, so this
definition coincides with previously introduced notions if one of the elements is a projection.
Note also that we do not use the word ''compatible'' in this case,
 since in general this notion is different from Mackey compatibility.

Assume now that $E$ has the b-property, so all elements are b-elements. By Lemma \ref{le:belement}, any element of $E$
is contained in some C-block of $E$.

\begin{theorem}\label{thm:cblock} If $E$ has the b-property, then
C-blocks in $E$ coincide with maximal sets of pairwise commuting elements in $E$.

\end{theorem}

\begin{proof}
 First observe that $PC(B)=B$ by maximality of $B$.
Let $a,b\in C(B)$, then $P(a),P(b) \subseteq C(p)$  for all $p\in B$, therefore
$P(a),P(b)\subseteq PC(B)= B$. This implies that $P(a)\leftrightarrow P(b)$, i.e. $aCb$. Let $g\in E$ be such that $gCa$
for all $a\in C(B)$. In particular,  $g\leftrightarrow p$ for all $p\in B$, which means that $g\in C(B)$. This implies that
$C(B)$ is a maximal set of mutually commuting elements.

Conversely, let $C\subset E$ be maximal with respect the property $aCb$ for all $a,b\in C$. This means that
$P(a)\leftrightarrow P(b)$ for all $a,b\in C$, hence there exists a block $B$ of $P$ with $\bigcup_{a\in C}P(a)\subseteq B$.
This entails that $P(a)\leftrightarrow B$, which implies  $a\in C(B)$  for all $a\in C$, hence $C\subseteq C(B)$. By the
first part of the proof and
 maximality of $C$, we have $C=C(B)$.
\end{proof}

\begin{example}\label{ex:bprop_quantum} Let $E=E(\mathcal H)$. Since for $a\in E$, $\{a\}'=\{a\}'''=P(a)'$ (see Example \ref{ex:hilb_commutant}), we see that $E$ has the b-property and for $a,b\in E$, we have $aCb$ if and only if
$ab=ba$.
Consequently, the C-blocks of $E(\mathcal H)$ indeed coincide with unit intervals in maximal abelian subalgebras in $B(\mathcal H)$.

\end{example}

\begin{definition}\label{de:b-compar} {\rm (Cf. \cite[Definition 6.3]{Pucompr}) } An effect algebra  $E$ has the
\emph{b-comparability property} if
\begin{enumerate}
\item[(a)] $E$ has  the b-property.
\item[(b)] For all $e,f\in E$ such that  $eCf$, the set
\[
P_\le(e,f):=\{p\in P(e,f):\ J_p(e)\le J_p(f)\\ \text{ and } J_{1-p}(f)\le J_{1-p}(e)\}
\]
is nonempty.
\end{enumerate}
\end{definition}

To provide some intuition for the definition  of b-comparability, we look at the spectral
resolution of a self-adjoint element  $a\in B(\mathcal H)$. We may restrict to a commutative
subalgebra containing $a$, which is isomorphic to a space of functions. Spectral projections
of $a$ correspond to subsets on which the function representing $a$ is comparable with a
multiple of the unit. To find such projections, we need to be able to find a decomposition
of elements in the subalgebra to a positive and negative part, which can be seen to be  equivalent to the property (b) in 
Definition \ref{de:b-compar} for pairs of commuting self-adjoint elements. See also Example \ref{ex:hilb_commutant}.

The b-comparability property has important consequences on the set of projections and on the structure of the C-blocks.

\begin{theorem}\label{th:proper}{\rm \cite[Thm. 6.1]{Pucompr}} Let $E$  have the b-comparability property. Then
every sharp element is a projection: $P=E_S$.
%{\rm(ii)} The compression base in $E$ is proper.
\end{theorem}

Note that by the above theorem, if b-comparability holds, then the compression base must be maximal.

\begin{theorem}\label{th:bcompar} { \rm  \cite[Thm. 7.1]{Pucompr}} Let $E$ have the b-comparability property and
 let $C=C(B)$ for a block $B$ of $P$. Then
\begin{enumerate}
\item $C$ is an MV-effect algebra.
\item For $p\in B$,  the restriction ${J}_p|_{C}$ coincides with $U_p$ (recall Example \ref{example:rdp}) and
$(U_p)_{p\in B}$ is the total compression base in $C$. Moreover, $C$ with $(U_p)_{p\in B}$
 has the b-comparability property.
\item If $E$ has the projection cover property, then $C$ has the projection cover property.
\item If $E$ is $\sigma$-orthocomplete, then $C$ is $\sigma$-orthocomplete.
\end{enumerate}
\end{theorem}

\begin{remark}\label{rem:b_prop} It may seem that we could have defined a notion of
comparability without assuming the b-property, requiring existence of some projection such
that $J_p(a)\le J_p(b)$ and $J_{p'}(b)\le J_{p'}(a)$ for all pairs of elements
$a,b\in E$. As was shown in \cite[Prop. 6.2]{Pucompr}, this would imply that all of $E$ is just one
C-block, hence an MV-effect algebra.

\end{remark}

\begin{example}\label{example:rdp2} Let $E$ be an effect algebra with RDP, see Example \ref{example:rdp}. Since all
elements are mutually compatible, all $E$ is one C-block, $E=C(E_S)$. By Theorem \ref{th:bcompar} (i), we see that if $E$
has the b-comparability property then $E$ must be an MV-effect algebra. As we will see in the next paragraph,
b-comparability in this case is equivalent to comparability in interpolation groups, see \cite{Good}.

\end{example}

\begin{example}\label{ex:center2} Let $E$ be an effect algebra and let $(U_p)_{p\in \Gamma(E)}$ be the central compression base
 (Example \ref{ex:center}). By Theorem \ref{th:proper}, if b-comparability holds then we must have $E_S=\Gamma(E)$,
that is, every sharp element is central. Again, under comparability, all of $E$ becomes one C-block $E=C(E_S)$, which is an
MV-effect algebra. Note that $E_S=\Gamma(E)$ e.g. in the case of \emph{effect monoids} (see
\cite{www}), that is, effect algebras endowed with a binary operation $\cdot: E\times E\to E$ which is
unital, biadditive and associative. For  $p\in E_S$ and $a\in E$, we have
\[
p\cdot a=a\cdot p=p\wedge a=U_p(a).
\]

\end{example}

Some further examples will be treated in Section \ref{sec:spect}.

\subsection{b-comparability and RDP}\label{sec:rdp}

%Let $M$ be an MV-effect algebra. Then an element $a\in M$ is sharp if and only if it is idempotent, that is,
%$a\dotplus a=a$. The set $M_S$ of sharp elements is a Boolean subalgebra in $M$ and it is often denoted by $B(M)$.

Let $E$ be an effect algebra with RDP.
By \cite{Rav},  $E$  is isomorphic to the unit interval in an abelian interpolation group $(G,u)$ with
order unit $u$ \cite{Good}, called the \emph{universal group} of $E$.
If moreover $E$ is lattice ordered (i.e. an MV-effect algebra), then $G$ is an $\ell$-group,  \cite{Mun}.
 In the rest of this section, we will identify  $E$ with the unit interval $[0,u]$ in its universal group.

Let us now recall the general comparability property in interpolation groups with order unit, see \cite[Chap. 8]{Good} for more details.
 For any sharp element $p\in E_S$, the convex subgroup $G_p$ generated by $p$ is an ideal of $G$ and
by RDP, we have the direct sum  $G=G_p\oplus G_{p'}$ (as ordered groups). Let $\tilde U_p$ be the projection of $G$
onto $G_p$ with kernel $G_{p'}$, then
\[
\tilde U_p(x)=x\wedge np,\quad \text{if }0\le x\le np \text{ for some } n\in \mathbb N.
\]
 Obviously, $\tilde U_p$ is an extension of the compression $U_p$ on
$E\simeq [0,u]$ (see Example \ref{example:rdp}). Note that the uniqueness of retractions implies that $(G,u)$ is a \emph{compressible group} in the
sense of \cite{Forc}.
%Any retraction $J_p:M\to M$ extends to a group homomorphism $G\to G$ (also denoted by $J_p$) determined
%on positive elements as
%\begin{equation}\label{eq:group_retr}
%J_p(a)=a\wedge mp,\qquad a\in G^+,\ a\le mu.
%\end{equation}
%By \cite[Prop. 8.3]{Good}, this extension can be described as follows. For a sharp element $p\in [0,u]$,
%let  $G_p$ be the order ideal in $G$ defined by $p$. Then $G$ is the direct sum $G=G_p\oplus G_{u-p}$  and $J_p$ is the
%projection onto $G_p$ with kernel $G_{u-p}$.
%
%Conversely, if $(G,u)$ is a lattice ordered abelian group with order unit $u$, then $M:=[0,u]$ is an MV-algebra and any map
% of the form \eqref{eq:group_retr} restricts to a retraction on $M$. .
%
We say that $(G,u)$ satisfies \emph{general comparability} if for any $x,y\in G$, there is some sharp element $p\in [0,u]$
such that $\tilde U_p(x)\le \tilde U_p(y)$ and $\tilde U_{p'}(x)\ge \tilde U_{p'}(y)$.
It is easily seen that $[0,u]$ has the
b-comparability property if $(G,u)$ satisfies general comparability. The aim of the rest of this section is to show that the converse is also true.

As noticed in Example \ref{example:rdp2}, if  $E=[0,u]$ has the b-comparability property, then it must be an MV-effect algebra. Consequently, the group
$G$ is lattice ordered. Let $a\in G$ be any element, then we have
\begin{equation}\label{eq:jordan_mv}
a=a_+-a_-,\qquad a_+,a_-\in G^+,\ a_+\wedge a_-=0,
\end{equation}
with  $a_+= a\vee 0$ and  $a_-=-a\vee 0$.
% We clearly have  $a_+-a\ge 0$, $a_+-a\ge -a$. Let $b\ge 0$, $b\ge -a$, then $b+a\ge 0$, $b+a\ge a$, so that $b+a\ge a\vee 0=a_+$, hence
% $b\ge a_+-a$. It follows that $a_+-a=(-a)\vee 0=a_-$.
%We have $a_+\wedge a_-=(a\vee 0)\wedge (-a\vee 0)=(a\wedge -a)\vee 0$. Now note that if $b\le a$ and $b\le -a$, then
% $2b\le 0$, so that $b\le 0$. It follows that $a\wedge -a\le 0$ and hence $a_+\wedge a_-=0$.
%

\begin{lemma}\label{lemma:MV_comp_intervals} Let $(G,u)$ be a lattice ordered abelian group with order unit an let
 $n\in \mathbb N$. The following are equivalent.
\begin{enumerate}
\item For any $a,b\in [0,nu]$ there is some sharp element $p\in [0,u]$ such that $\tilde U_p(a)\le \tilde U_p(b)$ and $\tilde U_{p'}(a)\ge
\tilde U_{p'}(b)$.
\item For any $a\in [-nu,nu]$ there is some sharp element $p\in [0,u]$ such that $\tilde U_p(a)\le 0$ and $\tilde U_{p'}(a)\ge 0$.
\end{enumerate}

\end{lemma}

\begin{proof} Assume (i) and let $a\in [-nu,nu]$. Let $a=a_+-a_-$ as in \eqref{eq:jordan_mv}, then clearly $a_+,a_-\in
[0,nu]$. By (i), there is
 some sharp element  $p\in [0,u]$ such that $\tilde U_p(a_+)\le \tilde U_p(a_-)$ and $\tilde U_{p'}(a_+)\ge \tilde U_{p'}(a_-)$. But then $\tilde U_p(a_+)=a_+\wedge
np\le a_+$ and also $\tilde U_p(a_+)\le \tilde U_p(a_-)\le a_-$, hence $\tilde U_p(a_+)=0$. Similarly, we obtain $\tilde U_{p'}(a_-)=0$. It follows
that
\begin{align*}
\tilde U_p(a)=\tilde U_p(a_+-a_-)=-\tilde U_p(a_-)\le 0,\qquad  \tilde U_{p'}(a)=\tilde U_{p'}(a_+-a_-)=\tilde U_{p'}(a_+)\ge 0.
\end{align*}
Conversely, if $a,b\in [0,nu]$, then $a-b\in [-nu,nu]$, so that (ii) clearly implies (i).

\end{proof}

\begin{theorem}\label{thm:cop_group} Let $(G,u)$ be an abelian interpolation  group with order unit $u$. Then $(G,u)$ satisfies general
comparability if and only if the effect algebra $[0,u]$ has the b-comparability property.

\end{theorem}

\begin{proof} Assume that $E=[0,u]$ has the b-comparability property, then $E$ is an MV-effect algebra and $G$ is
lattice ordered.  For $a,b\in G$,  there is some $m\in \mathbb N$
such that $a,b\in [-mu,mu]$. Then $a-b\in [-2mu,2mu]$, so it is clearly enough to prove that for each $n\in
\mathbb N$ and $a\in [-nu,nu]$, there is some sharp $p\in [0,u]$ such that $\tilde U_p(a)\le 0$ and $\tilde U_{p'}(a)\ge 0$.

We will proceed by induction on $n$. For $n=1$ the statement follows by Lemma \ref{lemma:MV_comp_intervals} and
the b-comparability property. So assume
that the statement holds for $n$, we will prove it for $n+1$. Using Lemma \ref{lemma:MV_comp_intervals}, we need to show
that for $a,b\in [0,(n+1)u]$ there is some sharp $p$ with $\tilde U_p(a)\le \tilde U_p(b)$ and $\tilde U_{p'}(a)\ge \tilde U_{p'}(b)$.
Notice that $a-u,b-u\in [-nu,nu]$, so that by the induction assumption, there are some sharp elements $q,r$ such that
\begin{align*}
\tilde U_q(a-u)\ge 0,\quad \tilde U_{q'}(a-u)\le 0,\qquad  
\tilde U_r(b-u)\ge 0,\quad \tilde U_{r'}(b-u)\le 0.
\end{align*}
This implies that
\[
\tilde U_q(a)\ge q,\ \tilde U_{q'}(a)\le q',\qquad \tilde U_r(b)\ge r,\ \tilde U_{r'}(b)\le r'.
\]
Consider the decompositions
\begin{align*}
a= \tilde U_{r\wedge q}(a)+\tilde U_{r'\wedge q}(a)+ \tilde U_{r\wedge q'}(a)+ \tilde
U_{r'\wedge q'}(a),\qquad
b= \tilde U_{r\wedge q}(b)+\tilde U_{r'\wedge q}(b)+ \tilde U_{r\wedge q'}(b)+ \tilde U_{r'\wedge q'}(b).
\end{align*}
We have $\tilde U_{r\wedge q}(a)-r\wedge q=\tilde U_{r\wedge q}(a-u)\in [0,nu]$ and similarly $\tilde U_{r\wedge q}(b)-r\wedge q\in [0,nu]$.
 By the induction hypothesis and Lemma \ref{lemma:MV_comp_intervals}, there is some sharp $s$ such that
\begin{align*}
\tilde U_s(\tilde U_{r\wedge q}(a)-r\wedge q)\le \tilde U_s(\tilde U_{r\wedge
q}(b)-r\wedge q), \qquad
\tilde U_{s'}(\tilde U_{r\wedge q}(a)-r\wedge q)\ge \tilde U_{s'}(\tilde U_{r\wedge q}(b)-r\wedge q).
\end{align*}
It follows that $\tilde U_{s\wedge r\wedge q}(a)\le \tilde U_{s\wedge r\wedge q}(b)$ and $\tilde U_{s'\wedge r\wedge q}(a)\ge \tilde U_{s'\wedge
r\wedge q}(b)$. Further, we have
\[
\tilde U_{r'\wedge q}(a)=\tilde U_{r'}(\tilde U_q(a))\ge r'\wedge q \ge \tilde U_q(r')\ge \tilde U_q(\tilde U_{r'}(b))=\tilde U_{r'\wedge q}(b)
\]
and similarly
\[
\tilde U_{r\wedge q'}(a)\le \tilde U_{r\wedge q'}(b).
\]
Finally, we have $\tilde U_{r'\wedge q'}(a)\le \tilde U_{r'}(q')=r'\wedge q'\le u$ and $\tilde U_{r'\wedge q'}(b)\le r'\wedge q'\le u$.
 Using the general comparability property in $[0,u]$, there is some sharp element $t$ such that
\[
\tilde U_{t\wedge r'\wedge q'}(a)\le \tilde U_{t\wedge r'\wedge q'}(b), \ \  \tilde U_{t'\wedge r'\wedge q'}(a)\ge \tilde U_{t'\wedge r'\wedge
q'}(b).
\]
Now put
\begin{align*}
p:=s\wedge r\wedge q + r\wedge q'+t\wedge r'\wedge q',\qquad p':=s'\wedge r\wedge q+ r'\wedge q + t'\wedge r'\wedge q'.
\end{align*}
Then we have
\begin{align*}
\tilde U_p(a)=\tilde U_{s\wedge r\wedge q}(a)+\tilde U_{r\wedge q'}(a)+ \tilde U_{t\wedge r'\wedge q'}(a)
\le \tilde U_{s\wedge r\wedge q}(b)+\tilde U_{r\wedge q'}(b)+ \tilde U_{t\wedge r'\wedge q'}(b)= \tilde U_p(b)
\end{align*}
and similarly also $\tilde U_{p'}(a)\ge \tilde U_{p'}(b)$. This proves the 'if' part. The converse is obvious.

\end{proof}

Recall that if the MV-effect algebra $E$ is archimedean, then it is isomorphic to a subalgebra of continuous functions
on a compact Hausdorff space $X$, \cite[Thm. 7.1.3]{DvPu}. If $E$ has b-comparability, then we have seen that the universal group has general
comparability and by Lemma \ref{lemma:app_arch}, the group $G$ is archimedean as well. Using the results of \cite[Chap.
8]{Good}, we obtain more information on the representing space $X$.

\begin{corollary}\label{coro:MV_comp} Let $E$ be an archimedean MV-effect algebra with the  b-comparability
property. Then $E$ is isomorphic to a subalgebra of continuous functions $X\to [0,1]$ for a totally disconnected compact
Hausdorff space $X$. The space $X$ is the Stone space of the Boolean subalgebra $P=E_S$ of sharp elements in $E$.

\end{corollary}

\subsection{The positive part and splitting}

Assume that $E$ has the b-comparability property and let $a,b\in E$, $aCb$. Then there is some C-block $C$ in $E$ such that $a,b\in C$ and therefore also $P_\le(a,b)\subseteq C$.
 By Theorem \ref{thm:cblock}, $C$ is an MV-effect algebra with b-comparability property and hence
its universal group $G$  has general  comparability, by Theorem
 \ref{thm:cop_group}. The group $G$ contains the element $b-a$ and any $p\in P_\le(a,b)$ is also
 contained in $P^G_\pm(b-a)$. By Lemma \ref{lemma:orthog_decomp}, for any $p\in P_\le(a,b)$, the positive part of $b-a$ has the form
 \[
 (b-a)_+ =\tilde U_p(b)-\tilde U_p(a)=J_p(b)\ominus J_p(a)\in C,
 \]
 so that  $(b-a)_+$ is well defined as an element in $E$.

\begin{lemma}\label{lemma:bminusa_props} Assume that $E$ has the b-comparability property and let $a,b\in E$, $aCb$.

\begin{enumerate}
\item The element $(b-a)_+$ is  contained in any C-block containing $a$ and $b$.
\item For $q\in P_\le(a,b)$,
$(a-b)_+=J_{q'}(a)\ominus J_{q'}(b)$.
\item Let $r\in PC(a,b)$ be such that $J_r(a)\le J_r(b)$ and $J_{r'}(b)\le
J_{r'}(a)$, then $J_r(b)\ominus J_r(a)=(b-a)_+$.

\item If $(b-a)_+\dg$ exists, then  $(b-a)_+\dg\in P_\le(a,b)$ and  it is the smallest projection such that $r\in PC(a,b)$
and $J_r(a)\le J_r(b)$, $J_{r'}(b)\le J_{r'}(a)$.

\end{enumerate}

\end{lemma}

\begin{proof}
The statement (i) is obvious from the above considerations, (ii) follows easily from the fact that $q'\in P_\le(b,a)$ if $q\in P_\le(a,b)$.
The statement (iii) follows from the fact that there is some C-block containing $a,b$ and $r$ and the properties of the orthogonal decomposition of $a-b$ in its universal group (Lemma \ref{lemma:orthog_decomp}).
 For (iv), note that
since the projection cover $p:=(b-a)_+\dg\in P((b-a)_+)$, (i) implies that $p$ is contained in any C-block containing $a,b$. This implies $p\in P(a,b)$ and using Lemma \ref{lemma:orthog_decomp} (ii) in the universal group of any such C-block, $p\in P_\le(a,b)$. The last part of (iv) follows again from Lemma \ref{lemma:orthog_decomp} (ii) using the fact that there is a C-block containing $a,b$ and $r$.

\end{proof}

For any element $a\in E$, we clearly have $aCa'$. Let $q\in P_\le(a,a')$, then $q\in P(a)$ and we have
\begin{align}
(a'-a)_+&=J_q(a')\ominus J_q(a)=q\ominus 2J_q(a)\label{eq:split1}\\
(a-a')_+&=J_{q'}(a)\ominus J_{q'}(a')=q'\ominus 2J_{q'}(a').\label{eq:split2}
\end{align}
In other words, we have a decomposition $a=J_q(a)\oplus J_{q'}(a)=a\wedge q\oplus
a\wedge q'$ where $2(a\wedge q)$ and $2(a'\wedge q')$ exist.  This may be interpreted as  ''$a\wedge q\le 1/2$''
although of course the constant $1/2$ may not be defined in $E$. Similarly, ''$a\wedge
q'\ge 1/2$''. We will call this decomposition the \emph{splitting} of $a$ and $q$ a
\emph{splitting projection} for $a$. If $E$ is archimedean then $C$ is archimedean as well and hence  is isomorphic to a subalgebra of continuous functions on a
compact Hausdorff space (Corollary \ref{coro:MV_comp}). In this case the splitting of $a$ has a clear interpretation.

\begin{lemma}\label{lemma:splitting} Let $a\in E$ and let $C$ be any C-block containing
$a$, with universal group $G$.
\begin{enumerate}
\item[(i)] For any  splitting projection $q$ for $a$, all the elements
$(a-a')_+$, $(a'-a)_+$, $q$, $2J_q(a)$, $2J_{q'}(a')$ are contained in  $C$
 and $(a-a')_+=(2a-1)_+$, the latter element computed in $G$.
 \item[(ii)] If $q$ is a splitting projection for $a$, then $q'$ is a splitting projection
 for $a'$.
\item[(ii)] If $(a-a')_+\dg$ exists, then $((a-a')_+\dg)'$ is the largest splitting projection for $a$.
\end{enumerate}

\end{lemma}

\begin{proof} Straightforward from Lemma \ref{lemma:bminusa_props}.

\end{proof}

Let $q\in P$, then $J_q(E)=[0,q]$ is an effect algebra, with addition from $E$ and with unit $q$. Note that the orthosupplement of an element $a$ in $[0,q]$ is
$a'_q:=q\ominus a=q\wedge a'$.

As proved in \cite[Thm. 4.4]{Gucomprba},  for $p\le q$, the restriction of $J_p$ is a compression on $[0,q]$ and $(J_p)_{P\cap [0,q]}$ is a compression base in $[0,q]$. Below we will always assume this inherited compression base in $[0,q]$.
  For $r\in P\cap [0,q]$, we will denote by $C_q(r)$ the commutant of $r$ and by $P_q(a)$ the bicommutant of $a$ in $[0,q]$.

\begin{lemma}\label{lemma:subinterval} Let $q\in P$ and  $a,r\in [0,q]\cap  P$.
\begin{enumerate}
\item[(i)] $C_q(r)=C(r)\cap [0,q]$ and $\{p\wedge q:\ p\in P(a)\}\subseteq P_q(a)$.
\item[(ii)] If $p\leftrightarrow q$, then $a\in C_q(p\wedge q)$ if and only if $a\in C(p)$.

\item[(iii)] If $E$ has b-comparability then $[0,q]$ has b-comparability.
\item[(iv)] Let   $[(a-a'_q)_+]_q$ and $[(a'_q-a)_+]_q$ be the splittings in $[0,q]$, then
\begin{align*}
[(a-a'_q)_+]_q=(a-a')_+\wedge q,\qquad [(a'_q-a)_+]_q=(a'-a)_+\wedge q.
\end{align*}
\item[(v)] Let $[a\dg]_q$ be the projection cover in $[0,q]$, then $[a\dg]_q=a\dg$. For
$b\in C(q)$, $(b\wedge q)\dg=b\dg\wedge q$.
\end{enumerate}

\end{lemma}

\begin{proof} For (i), let $b=J_q(b)$, then $J_r(b)\oplus J_{r'_q}(b)=J_r(b)\oplus J_{r'\wedge q}(b)=J_r(b)\oplus J_{r'}(b)$, this implies the first statement in (i) and also that
$PC_q(a):=\{s\in P\cap [0,q], a\in C_q(s)\}=PC(a)\cap [0,q]$.
Let $p\in P(a)$, then $p$ commutes with $q$ and $p\wedge q$ exists. For the second statement of (i), we have to prove that $p\wedge q\in C_q(PC(a)\cap [0,q])$ which amounts to showing that $p\wedge q\in C(PC(a)\cap [0,q])$. Since clearly $p,q\in C(PC(a)\cap [0,q])$, we have the same for $p\wedge q$.

For (ii), it is enough to note that if $p\leftrightarrow q$, then $(p\wedge q)'_q=q\wedge(p\wedge q)'=q\wedge p'$ and hence
\[
J_{p\wedge q}(a)\oplus J_{(p\wedge q)'_q}(a)=J_p(a)\oplus J_{p'}(a).
\]

For (iii), let us prove the b-property. By (i), it is enough to show that  $P_q(a)\subseteq  C(r)$ if and only if $a\in C(r)$. So assume that $P_q(a)\subseteq C(r)$, then for any $p\in P(a)$,we have $p\wedge q\in C(r)$ which by (ii) (applied to $a=r$) implies that $p\in C(r)$. By the b-property in $E$, $a\in C(r)$. The converse is clear by definition of $P_q(a)$, this shows the b-property in  $[0,q]$.  If $a$ commutes with $b$ in $[0,q]$ then by (i) we must have
$P(a)\wedge q\leftrightarrow P(b)\wedge q$. Using (ii), we obtain that for all $p\in P(a)$, $s\in P(b)$ we must have $p\in C(s\wedge q)$, and since  $a\in C(s\wedge q')$, $p$ commutes
with $s\wedge q'$ as well. It follows that $P(a)\leftrightarrow P(b)$ and
 $aCb$ in $E$. Hence  there is some $t\in P_\le(a,b)$, it is easily seen that we may use the projection
$t\wedge q$ to obtain b-comparability in $[0,q]$.

For (iv), it is easily checked that  for any $p\in
P_\le(a',a)$ we have $p\wedge q\in P_{q,\le} (a'_q,a)$ so that
\begin{align*}
[(a-a'_q)_+]_q&=J_{p\wedge q}(a)\ominus J_{p\wedge q}(a'_q)=J_q(J_p(a)\ominus
J_p(1-a))=J_q((a-a')_+)=(a-a')_+\wedge q.
\end{align*}
The proof for $[(a'_q-a)_+]_q$ is the same.

The projection cover in  $[0,q]$ is clearly the same as in $E$, we only have to prove the
second statement of (v).
It is clear that $(b\wedge q)\dg\le b\dg\wedge q$.
 Conversely, since  $b\dg\in C(q)$, we have
 \[
b=b\wedge q\oplus b\wedge q'\le (b\wedge q)\dg\oplus (b\wedge q')\dg\le b\dg \wedge
q\oplus b\dg \wedge q'=b\dg.
 \]
Since $b\dg$ is the smallest projection majorizing $b$, the second inequality must be an
equality, this implies $(b\wedge q)\dg=b\dg \wedge q$.

\end{proof}

\section{Spectral effect algebras}

We are now ready to introduce a notion of spectrality in effect algebras.

\label{sec:spect}
\begin{definition}\cite{Pucompr} Let $E$ be an effect algebra with compression base $(J_p)_{p\in P}$. We say that $E$ is spectral if $(J_p)_{p\in P}$ has both the b-comparability and projection cover property.

\end{definition}

Let us remind the reader about some consequences of this definition. First, $E$ is endowed
with a compression base $(J_p)_{p\in P}$ such that the set of projections $P=E_S$. This
implies that the compression base is maximal and must contain the central compression base
$(U_p)_{p\in \Gamma(E)}$. Another consequence is that the set of sharp elements $E_S$ is a
normal subalgebra which is an OML with the structure inherited from $E$. Every element
$a\in E$ is contained in a C-block $C\subseteq E$, consisting of all elements compatible
with all sharp elements in some block $B$ of $P$. Any C-block is an MV-effect algebra that is spectral with respect to its
unique maximal compression base $(U_p)_{p\in C_S}$.  Spectral MV-effect algebras will be treated (among others) in the examples below. It will be also shown (Example \ref{ex:OMP_spec}) that the covering by spectral MV-effect algebras is in general not enough for the whole effect algebra $E$ to be spectral.

Note that most of the above properties are obtained from b-comparability. Without the
projection cover property we would have a similar structure, but the subalgebra $P=E_S$
might not be a lattice and the MV-effect algebras might not be spectral. Notice also that
 our construction of the spectral resolution below, based on the splitting,  still may be done, 
 but with no unique choice for the spectral projections (cf. Lemma \ref{lemma:app_sum}
 below or \cite[Thm. 3.22]{JenPul} for order unit spaces). 

\begin{example} \label{ex:hilb_spec} Let $E=E(\mathcal H)$ for some Hilbert space $\mathcal H$ (or more generally, we may assume that
$E$ is the interval $[0,1]$ in any von Neumann algebra or in a JBW algebra), with the compression base $(U_p)_{p\in P(\mathcal H)}$
defined as in Example \ref{ex:compr_effects}. For any $a\in E$, let $a\dg$ be the support projection of $a$, then
$a^{\circ}$ is a projection cover of $a$. As we have seen in Example \ref{ex:bprop_quantum}, $E$ has the b-property and $aCb$ if and only if
$ab=ba$. In that case, taking the positive part of the self-adjoint operator $a-b$, we have $(a-b)_+\dg\in P_\le(b,a)$. We conclude that $E$ has both the projection cover and
b-comparability property, so that $E$ is spectral.

\end{example}

\begin{example}\label{ex:MV} Let $E$ be an MV-effect algebra and let $(G,u)$ be the universal group.
By Theorem \ref{thm:cop_group}, $E$ has the b-comparability property if and only if $G$ has the comparability property. In this case, the projection cover property in $E$ is equivalent to the \emph{Rickart property} in $G$,  see \cite[Thm. 6.5]{Fgc} and Appendix \ref{app:spec}. Hence $E$ is spectral if and only if $G$ is spectral.  Assume further that $E$ is monotone $\sigma$-complete, then $(G,u)$ is a Dedekind $\sigma$-complete lattice ordered group with order unit.
By \cite[Thm. 9.9]{Good}, $(G,u)$ satisfies general comparability, so $E$ has the b-comparability property. Moreover, it follows by
\cite[Lemma 9.8]{Good} that $E$ has the projection cover property, so that $G$ and also $E$ is spectral. We will see in Example \ref{ex:convexMV} that under additional conditions a spectral MV-algebra must be monotone $\sigma$-complete.  However, this is not always the case: as we will see next, any Boolean algebra is spectral.

\end{example}

\begin{example}\label{ex:OMP_spec} Let $E$ be a Boolean algebra, with the central compression base $(U_p)_{p\in E}$. It is easily seen that  $E$ has b-comparability: the b-property is obtained by setting $B(q)=\{0,q,q',1\}$, moreover,   for $p,q\in E$, we have $P(p,q)=E$ and
\begin{align*}
U_q(p)=q\wedge p\le q=U_q(q),\qquad  
U_{q'}(q)=0\le q'\wedge p=U_{q'}(p),
\end{align*}
so that $q\in P_\le (p,q)\ne \emptyset$. It is also clear that any Boolean algebra has the projection cover property, by
setting $p\dg=p$, $p\in E$. Hence any Boolean algebra is spectral.

More generally, let $E$ be an OMP. Since in this case $E=E_S$, we see by Theorem \ref{th:proper} that
 if $E$ has b-comparability, then $P=E_S=E$, so that $E$ must be a Boolean algebra (see Example \ref{ex:omp}). Note that any OMP is covered by blocks which are Boolean algebras, hence spectral MV-effect algebras, but is itself not spectral unless it is Boolean.

\end{example}

\begin{example}\label{ex:directpr} Let $E_1$, $E_2$ be effect algebras with compression bases $(J_{1,p})_{p\in P_1}$ and $(J_{2,p})_{p\in P_2}$.
Let $E=E_1\times E_2$ be the direct product and let $P=P_1\times P_2$. For $p=(p_1,p_2)\in P$, put $J_p=(J_{1,p_1},J_{2,p_2})$. It is easily checked that each $J_p$ is a compression with focus $p$ and $(J_p)_{p\in P}$ is a compression base. Since all the operations are taken pointwise, it can be checked that $E$ is spectral if and only if both $E_1$ and $E_2$ are spectral.

\end{example}

\begin{example}\label{ex:horsum_spectral}
 Let $E=E(\mathcal H)\dot{\cup} E(\mathcal H)$ be the horizontal sum as in Example \ref{ex:horsum}, with the compression base $(J_p)_{p\in P}$, $P=P(\mathcal H)\dot{\cup}P(\mathcal H)$,  constructed from faithful states on $E(\mathcal H)$. Since any element of $E$ can be compatible with a projection only if both belong to the same component, it is easily seen that $E$ is spectral. More generally, if $E_1$ and $E_2$ are effect algebras with compression bases $(J_{1,p})_{p\in P_1}$ and $(J_{2,p})_{p\in P_2}$ such that we can construct a compression base with $P=P_1\dot{\cup} P_2$ as in Example \ref{ex:horsum}, then $E_1\dot{\cup} E_2$ is spectral if and only if $E_1$ and $E_2$ are spectral.

\end{example}

All the properties involved in the definition of spectrality depend on the choice of the
compression base. As remarked above, such a compression base must satisfy $P=E_S$ and
therefore must be maximal. We have seen in Example \ref{ex:horsum} and
\ref{ex:horsum_spectral} that such a compression base may be not unique. The next result
shows that spectrality does not depend from the choice of such a compression base.

\begin{prop}\label{prop:compbase_unique} Let $E$ be spectral and let $\{\bar J_p\}_{p\in
P}$ be a compression base with $P=E_S$. Then $E$ is also spectral with respect to
$\{\bar J_p\}$.

\end{prop}

\begin{proof} By assumption, the compression bases $(J_p)$ and $(\bar J_p)$ have the
same set of projections. It is easily seen that the projection cover for any element is the same for
both compression bases, since these  only depends on the properties of the
projections and not the corresponding compressions. Hence $E$ with $\{\bar J_p\}$ has the
projection cover property. By Lemma \ref{lemma:compatible_projs}, we also see that the
commutants $C(p)$ and bicommutants $P(a)$ are the same for both compression bases, and
that for $a\in C(p)$, $J_p(a)=\bar J_p(a)=p\wedge a$. This shows that $E$ with $\{\bar
J_p\}$ has the b-comparability property.

\end{proof}

It would be interesting to know if this is true for any maximal compression base, in other
words, if for a spectral effect algebra $E$ any maximal compression base must satisfy
$P=E_S$.

\subsection{Spectral resolutions}

We will define the \emph{binary spectral resolution} of an element $a\in E$  as a family of projections
$\{p_\lambda\}$, indexed by dyadic rationals in $[0,1]$. A dyadic rational 
$\lambda\in [0,1)$ has an expansion
\begin{align*}
\lambda=\lambda(w):=\sum_{j=1}^{l(w)} w_j2^{-j}=k(w)/2^{l(w)},\qquad k(w):=\sum_{j=1}^{l(w)}w_j2^{l(w)-j}
\end{align*}
 for a binary string  $w=w_1\dots w_{l(w)}\in \mathcal B:=\bigcup_{n\ge 0} \{0,1\}^n$,
 here $l(w)$ is the length of $w$.  We will denote the empty string as $\varepsilon$ and the
 concatenation of $w_1,w_2\in \mathcal B$ as $w_1w_2$. The
lexicographical order on binary strings will be denoted by $\le$.

 Any dyadic rational $\lambda\in (0,1)$  has a unique expansion of the
 form  $\lambda=\lambda(w1)$ for some $w\in \mathcal B$, in this case
\[
\lambda=\lambda(w1)=(2k(w)+1)/2^{-(l(w)+1)}.
\]
Note that $\lambda(w1)$ is
the middle point of the interval $[\lambda(w),\lambda(w+1)]$ where $w+1$ is the consecutive binary
string of the same length as $w$ (if it exists, if $w$ is a (possibly empty) string of
ones $w=1\dots 1$, we put $\lambda(w+1)=1$). Similarly, we denote by $w-1$ the previous
binary string of the same length as $w$ and if $w$ is already the smallest element in
$\{0,1\}^{l(w)}$, then $\lambda(w-1)=0$.

We first construct two families $\{c_w\}$ and $\{u_w\}$ indexed by binary strings:
\begin{itemize}
\item we put $c_\varepsilon=a$, $u_\varepsilon:=a\dg$,
\item for any $w\in\{0,1\}^n$, $n=0,1,\dots$, we define $c_{w0}$ and $c_{w1}$ by the
splitting of $c_w$ in $[0,u_w]$:
\begin{align*}
u_{w0}:=&((c_w-c_w')_+\dg)'\wedge u_w,\qquad  u_{w1}:=u_w\wedge u_{w0}',\\
c_{w0}:=&u_{w0}\ominus [(c_w'-c_w)_+\wedge u_w]=2J_{u_{w0}}(c_w),\\
c_{w1}:=&(c_w-c_w')_+\wedge u_w =u_{w1}\ominus 2 J_{u_{w1}}(u_w\ominus c_w).
\end{align*}

\end{itemize}
Note that we have used the properties of splitting in \eqref{eq:split1}, \eqref{eq:split2} and Lemmas \ref{lemma:splitting}, \ref{lemma:subinterval}.

\begin{remark} It might be useful to write down the elements $u_w$ and $c_w$ explicitly in a simple
example. Let $E=[0,1]^n$, then $E$ is a monotone $\sigma$-complete $MV$-algebra, hence $E$
is spectral. Let $a=(a_1,\dots, a_n)\in E$. Note that an element $u\in E$ is sharp
precisely if $u_i\in \{0,1\}$ for all $i$, hence can be identified with a subset
$U\subseteq \{1,\dots, n\}$, $U=\{i, u_i=1\}$. Then we have, for any binary string $w$ of
length $n$:
\begin{align*}
U_w=\{i, \frac{k(w)}{2^n}< a_i \le \frac{k(w)+1}{2^n}\},\qquad   c_w(i)= \begin{dcases}
2^na_i-k(w), & i\in U_w\\
0,& i\ne U_w.
\end{dcases}
\end{align*}
See also Lemma \ref{lemma:cw} below.

\end{remark}

We collect some easily observed  properties of these families.

\begin{lemma}\label{lemma:uw}
\begin{enumerate}
\item[(i)] For all $w$, we have $u_w\in P(a)$, $c_w\le u_w$ and $c_w$ is contained in all
C-blocks containing $a$.
\item[(ii)] For any strings $w$, $\tilde w$, we have $u_{w\tilde w}\le u_w$.
\item[(iii)] For binary strings of fixed length $l(w)=n$, $u_w$ are pairwise orthogonal and we have
\[
\bigoplus_{w\in \{0,1\}^n} u_w=a\dg.
\]
\end{enumerate}

\end{lemma}

\begin{proof}
It is clear from the construction that $c_w\le u_w$ for all $w$.
By Lemma \ref{lemma:splitting}, all $u_{w0}$, $u_{w1}$, $c_{w0}$,
$c_{w1}$ are contained in any C-block containing $c_w$, which implies (i).

The statement (ii) can be proved by induction on the length of $\tilde w$: for any $w$, we
have $u_{w0}, u_{w1}\le u_w$ by construction. If the statement holds for any $\tilde w$ of
length $n$, then $u_{w\tilde wi}\le u_{w\tilde w}\le u_w$ for $i\in \{0,1\}$.

We use induction on the length of $w$ to prove (iii): by construction $u_{w0}\oplus
u_{w1}=u_w$ for any $w$, this also proves the statement for $n=1$. Let $w_1,w_2\in\{0,1\}^n$ and assume that the statement holds
for any $k<n$. To show orthogonality we only have to deal with the
situation when $w_1$ and $w_2$ have some prefix $\tilde w_1$ resp.
$\tilde w_2$, of length $k<n$ and such that $\tilde w_1\ne \tilde w_2$, in which case  by (ii),
$u_{w_1}\le u_{\tilde w_1}$, $u_{w_2}\le u_{\tilde w_2}$ and $u_{\tilde w_1}\perp
u_{\tilde w_2}$ by the induction assumption. Moreover,
\begin{align*}
\bigoplus_{w\in \{0,1\}^n} u_w=\bigoplus_{\tilde w\in \{0,1\}^{n-1}} (u_{\tilde w0}\oplus
u_{\tilde w 1}) =\bigoplus_{\tilde w\in \{0,1\}^{n-1}}u_{\tilde w}=a\dg.
\end{align*}
\end{proof}

We now define the binary spectral resolution of $a$:
\begin{itemize}
\item $p_0= (a\dg)'$,
\item for  $l(w)=n$, $n=0,1,\dots $, we put
\[
p_{\lambda(w1)}=(a\dg)' \oplus \bigoplus_{\tilde w\in \{0,1\}^{n+1}, \tilde w\le w0} u_{\tilde w}
\]
\item $p_1=1$.
\end{itemize}

\begin{lemma}\label{lemma:pw1} For any $w\in \{0,1\}^n$, we have
\begin{enumerate}
\item[(i)] $p_{\lambda(w1)}\wedge u_w=u_{w0}$,
\item[(ii)]  $(a\dg)'\oplus \bigoplus_{\tilde w\in \{0,1\}^{l(w)},\tilde w\le w} u_{\tilde
w}=p_{\lambda(w+1)}=p_{\lambda(w)}\oplus u_w$,
\item[(iii)] if $w_1,w_2$ are two binary strings such that $w_1\le w_2$, then
$p_{\lambda(w_1)}\le p_{\lambda(w_2)}$.
\end{enumerate}

\end{lemma}

\begin{proof} The statement (i) follows easily from the definition and the fact that
$u_w=u_{w0}\oplus u_{w1}$.
For (ii), notice first that if $w=1\dots 1$ is the largest element in $\{0,1\}^n$,
then by Lemma \ref{lemma:uw}, we have
\[
(a\dg)'\oplus\bigoplus_{\substack{\tilde w\in \{0,1\}^n\\ \tilde w\le w}} u_{\tilde
w}=1=p_1=p_{\lambda(w+1)},
\]
so we may assume $w\ne 1\dots 1$. This also shows that the statement holds for $l(w)=0$.
Assume it is true for strings of length $n$ and let $w$ be a binary string of length $n+1$.  If $w=w_10$ for
some $w_1\in \{0,1\}^n$, then we have by definition that the left hand side is equal to
$p_{\lambda(w_11)}=p_{\lambda(w+1)}$. Let now $w=w_11$, then notice that
\begin{align*}
\bigoplus_{\substack{\tilde w\in \{0,1\}^{n+1}\\ \tilde w\le w}} u_{\tilde
w}=\bigoplus_{\substack{\tilde w\in \{0,1\}^{n+1}\\\tilde w\le (w_1-1)1}} u_{\tilde
w}\oplus u_{w_10}\oplus u_{w_11} =\bigoplus_{\substack{\tilde w\in \{0,1\}^{n+1}\\ \tilde
w\le (w_1-1)1}} u_{\tilde
w}\oplus u_{w_1}
\end{align*}
and repeating the argument, we obtain that
\begin{align*}
(a\dg)'\oplus\bigoplus_{\substack{\tilde w\in \{0,1\}^{n+1}\\ \tilde w\le {w_11}}}
u_{\tilde w}=(a\dg)'\oplus\bigoplus_{\substack{\tilde
w\in \{0,1\}^n\\ \tilde w\le {w_1}}} u_{\tilde
w}
=p_{\lambda(w_1+1)}=p_{\lambda(w+1)},
\end{align*}
using the induction assumption and $\lambda(w0)=\lambda(w)$. The second equality is now
clear. If $w_1$, $w_2$ have the same length, the statement (iii) follows from
$p_{\lambda(w)}\le p_{\lambda(w+1)}$ which is obvious from (ii). Otherwise we may add a
string of zeros at the end of one of $w_1$, $w_2$, which does not change the order nor the
values of $p_{\lambda(w_1)}$ and $p_{\lambda(w_2)}$.

\end{proof}

The element  $a$ is contained in some C-block $C=C(B)$ of $E$, which is a spectral
MV-effect algebra with the compression base $(J_p|_C=U_p)_{p\in B}$.  As we have seen in Example
\ref{ex:MV}, the universal group $G$ of $C$ is a spectral lattice ordered abelian group $G$ with order unit, with the compression base $(\tilde U_p)_{p\in B}$, extending the compression base in $C$. Hence there
is a {rational spectral resolution} of $a$ in $G$, given by 
\begin{equation}\label{eq:spectprojs}
p^C_{a,\lambda}:= ((na-mu)_+)^*,\qquad \lambda=\frac{m}{n},\ n>0,
\end{equation}
see Appendix \ref{app:spec}. For $g\in G$, $g^*$  denotes the Rickart mapping in $G$.
The notation  highlights the  possible dependence on the choice of the C-block
$C$. If the
element $a$ is clear, we will skip the index from the notation. 

As it will turn out below, the spectral resolution does not depend on $C$.
We will first settle the result for dyadic rationals, relating the spectral projection
$p^C_\lambda$ to the corresponding projection $p_{\lambda(w1)}$ constructed above.

\begin{lemma}\label{lemma:cw} Let $w=w_1\dots w_{l(w)}$ and let $C=C(B)$ be a C-block
containing $a$. Then $u_w\in B$ and
\[
c_w=\tilde U_{u_w}(2^{l(w)}a-k(w)1),
\]
computed in the universal group $G$ of $C$.
\end{lemma}

\begin{proof} By Lemma \ref{lemma:uw}, we have  $u_w\in P(a)\subseteq C\cap P=B$ for all
$w$. For $w=\varepsilon$, we have $l(w)=k(\varepsilon)=0$ and
$c_\varepsilon=a=U_{a\dg}(a)=U_{u_\varepsilon}(a)=\tilde U_{u_\varepsilon}(a)$. Assume that the
equality holds for  $l(w)=n$. We have $k(w0)=2k(w)$ and by construction and induction assumption,
\begin{align*}
c_{w0}&=2U_{u_{w0}}(c_w)=2\tilde U_{u_{w0}}(c_w)=2\tilde U_{u_{w0}}(2^na-k(w)1)
=\tilde U_{u_{w0}}(2^{n+1}a-2k(w)1)\\
&=\tilde U_{u_{w0}}(2^{l(w0)}a-k(w0)1),
\end{align*}
here we used the fact that $u_{w0}\le u_w$ and that $\tilde U_p$ is  a group homomorphism
on $G$. Similarly, $k(w1)=2k(w)+1$ and
\begin{align*}
c_{w1}=\tilde U_{w1}(2c_w-1)=\tilde U_{w1}(2(2^{l(w)}a-k(w)1)-1)=\tilde U_{w1}(2^{l(w1)}a-k(w1)).
\end{align*}

\end{proof}

\begin{theorem}\label{thm:spectral} For any C-block $C$ such that  $a\in C$ and any dyadic
rational $\lambda\in [0,1]$  we have
\[
p^C_{\lambda}=p_{\lambda}.
\]

\end{theorem}

\begin{proof} It is easily seen from \eqref{eq:spectprojs}  that $p^C_{0}=1-a\dg=p_0$
and $p^C_{1}=1=p_1$. Let us now prove that $p^C_{\frac12}=p_{\lambda(1)}$. By definition, we
have
\[
p_{\lambda(1)}=(a\dg)'\oplus u_0=(a\dg)'\oplus ((a-a')_+\dg)'\wedge a\dg.
\]
Notice that, for any $p\in P_\le(a',a)$, we have
\begin{align*}
(a-a')_+=J_p(a)\ominus J_p(a')\le J_p(a)\le J_p(a\dg)=a\dg\wedge p\le a\dg,
\end{align*}
so that $(a-a')_+\dg\le a\dg$ and $((a-a')_+\dg)'\wedge (a\dg)'=(a\dg)'$. It follows that
\begin{align*}
p_{\lambda(1)}=((a-a')_+\dg)'\wedge(a\dg)'\oplus  ((a-a')_+\dg)'\wedge
a\dg=((a-a')_+\dg)'=p^C_{\frac12}.
\end{align*}
We will now show the statement by induction on $l(w)$ in the expansion
$\lambda=\lambda(w1)$,
the above paragraph shows that it is true for $l(w)=0$.

Assume that the statement holds if $l(w)<n$ and let $w$ be such that $l(w)=n$. We have
$\lambda(w)\le \lambda\le \lambda(w+1)$ and
by Lemma \ref{lemma:pw1}, $p_{\lambda(w)}\le
p_\lambda\le p_{\lambda(w+1)}$ and $u_w=p_{\lambda(w+1)}\wedge p_{\lambda(w)}'$.  By the induction assumption we have
$p_{\lambda(w)}=p^C_{\lambda(w)}$ and $p_{\lambda(w+1)}=p^C_{\lambda(w+1)}$, so that using
 the properties of spectral projections in Lemma \ref{lemma:pw1}, we obtain
\[
p_{\lambda(w)}\le p_{\lambda},p^C_{\lambda}\le p_{\lambda(w+1)}.
\]
This implies $p_{\lambda}\wedge u_w'=p_{\lambda}^C\wedge u_w'=p_{\lambda(w)}$, so that what
we have to prove is
\[
u_{w0}=p_{\lambda(w1)}\wedge u_w=p^C_{\lambda(w1)}\wedge u_w.
\]
Put $g_w:=2^{l(w)+1}a-(2k(w)+1)1\in G$, then
$p^C_{\lambda}=(g_w)_+^*$ by \eqref{eq:spectprojs}. Using Lemmas \ref{lemma:orthog_decomp}
(iii)
and \ref{lemma:cw}, and the definition of $u_{w0}$, we have
\begin{align*}
\tilde U_{u_{w0}}((g_w)_+)=\tilde U_{u_{w0}}\tilde U_{u_w}((g_w)_+)=\tilde U_{w0}((\tilde
U_{u_w}(g_w))_+)=J_{u_{w0}}((c_w-(u_w\ominus c_w))_+)=0.
\end{align*}
By the properties of the Rickart mapping, this implies that $u_{w0}\le
p^C_{\lambda}$ and since also $u_{w0}\le u_w$,  $u_{w0}\le p^C_{\lambda}\wedge u_w$.
The converse follows similarly from
\begin{align*}
J_{p_{\lambda}^C\wedge u_w}((c_w-(u_w\ominus
c_w))_+)=\tilde U_{p_{\lambda}^C\wedge u_w}((\tilde U_{u_w}(g_w))_+)=\tilde
U_{u_w}\tilde U_{p_\lambda^C}((g_w)_+)=0,
\end{align*}
the definition of $u_{w0}$ and properties of the projection cover.

\end{proof}

To obtain the result for values $\lambda\in \mathbb Q\cap [0,1]$ other than dyadic
rationals, it will be  convenient to use an additional assumption on $E$, namely that $E$ is archimedean. The next result shows that this assumption is necessary for an important right continuity property of the binary spectral resolution.

\begin{prop}\label{prop:right_cont} Let $E$ be spectral and let $\{p_{a,\lambda}\}$ be the binary spectral resolution of $a\in E$. Then
\[
p_{a,\lambda}=\bigwedge_{\mu>\lambda} p_{a,\mu}
\]
holds for any $a$ if and only if $E$ is archimedean.

\end{prop}

\begin{proof}  Let  $C$ be any C-block containing $a$ and let $G$ be its universal group. If
$E$ is archimedean, then $C$ is archimedean as well, so that $G$ is archimedean by Lemma
\ref{lemma:app_arch}.  By
Theorem \ref{thm:rsd}, we have for all $\lambda\in \mathbb Q\cap [0,1]$ that
\[
p^C_{a,\lambda}=\bigwedge_{\mu>\lambda} p^C_{a,\mu},
\]
where the infimum is over all $\mu\in \mathbb Q\cap [0,1]$.
Restricting to dyadic rationals, this implies the statement by Theorem \ref{thm:spectral}
and the fact that dyadic rationals are dense in $[0,1]$. Assume conversely that 
the equality holds and let $a\in E$ be such that $2^na$ exists for all $n\in \mathbb N$. Then for all $n$, we have
\[
p_{a,2^{-n}}=p^C_{a,2^{-n}}=(2^na-1)_+^*=1.
\]
Consequently,
\[
(a\dg)'=p_0=\bigwedge_n p_{a,2^{-n}}=1,
\]
so that $a=0$. Hence $E$ must be archimedean.
\end{proof}

\begin{theorem}\label{thm:spectral_archimedean}  Let $E$ be an archimedean spectral effect algebra and $a\in E$. Then $p^C_{a,\lambda}$ does not depend on the C-block $C$ for any $\lambda\in [0,1]\cap \mathbb Q$, in particular, we have
\[
p^C_{a,\lambda}=\bigwedge_{\mu>\lambda} p_{a,\mu}=:p_{a,\lambda},
\]
here the infimum is over all dyadic rationals $\mu>\lambda$.

\end{theorem}

\begin{proof} As in the proof of Proposition \ref{prop:right_cont}, we have  for  that
\[
p^C_{a,\lambda}=\bigwedge_{\lambda<\mu} p^C_{a,\mu}
\]
for any C-block $C$ containing $a$ and all $\lambda\in \mathbb Q\cap [0,1]$.  The
statement now follows from Theorem \ref{thm:spectral} and the fact that dyadic rationals  are dense in $[0,1]$.

\end{proof}

The family of projections $\{p_{a,\lambda}\}_{\lambda\in \mathbb Q\cap[0,1]}$ obtained in the
above theorem will be called the \emph{rational spectral resolution} of $a$ in $E$.

We introduce  partially defined maps $f_0,f_1$ on an effect algebra $E$, such that
\[
f_0(b)=2b,\qquad
f_1(b)=(2b')'
\]
defined if the respective elements exist, that is, $f_0(b)$ exists if $b\le b'$ and
$f_1(b)$ exists if $b'\le b$.  For a binary string $w=w_1\dots
w_n$ we define $f_w:=f_{w_n}\circ\dots\circ f_{w_1}$, with the understanding that
$f_\varepsilon$ is the identity map and if  $f_w(b)$ exists then
$f_{\tilde w}(b)$ must exist for all prefixes $\tilde w$ of $w$. It is easily seen that
for $p\in P$, $J_p(f_w(a))=J_p(f_w(J_p(a)))$ if $f_w(a)$ exists.

We have the following characterization of the rational spectral resolution.

\begin{theorem}\label{thm:spectral_char} Let $E$ be an archimedean spectral effect algebra.
Then the rational spectral resolution of $a\in E$ is the unique family
$\{p_\lambda\}_{\lambda\in \mathbb Q\cap[0,1]}$ such that
\begin{enumerate}
\item[(i)] $p_\lambda\in PC(a)$ for all $\lambda$,
\item[(ii)] $p_0\le a'$, $p_1=1$ and if $\lambda\le \mu$ then $p_\lambda\le p_\mu$,
\item[(iii)] $\bigwedge_{\lambda<\mu} p_\mu=p_\lambda$ for all $\lambda$,
\item[(iv)]  for any $w\in \mathcal B$ put $u_w:=p_{\lambda(w+1)}\wedge p_{\lambda(w)}'$
and let $f_w$ be defined as above for the effect algebra $[0,u_w]$, then   $f_w(J_{u_w}(a))$ exists.

\end{enumerate}

\end{theorem}

\begin{proof} Assume $\{p_\lambda\}_{\lambda\in \mathbb Q\cap[0,1]}$ is a family of
projections satisfying (i)-(iv). By (ii), all $p_\lambda$ commute, so that by (i) there is
some C-block $C$ containing $a$ and all $p_\lambda$, with universal group $G$.  We will show that (iv) is equivalent
to
\begin{enumerate}
\item[(iv')]  $2^{l(w)} J_{p_{\lambda(w)}}(a)\le k(w)p_{\lambda(w)}$ and
 $2^{l(w)}J_{p_{\lambda(w)}'}(a)\ge k(w)p'_{\lambda(w)}$, for all $w\in
\mathcal B$, computed in $G$.
\end{enumerate}
To show this, notice that for $w\in \mathcal B$, $i\in \{0,1\}$ and $b\in [0,u_w]$,  $f_i(b)$ exists in $[0,u_w]$ if and only if  $2b-iu_w$ computed in
$G$ belongs to $C$ and then $f_i(b)=2b-iu_w$.
We have by (ii) and (iv) that $u_{wi}\le u_w$ and in $[0,u_{wi}]$
\[
f_{wi}(J_{u_{wi}}(a))=f_i(J_{u_{wi}}(f_w(J_{u_w}(a)))).
\]
Using the fact that  $k(wi)=2k(w)+i$, one can easily prove by induction that (iv) is
equivalent to
\[
\tilde U_{u_w}(2^{l(w)}a-k(w)1)\in C,\qquad \forall w\in \mathcal B,
\]
 here $u_w$ is as in (iv) and $\tilde U_{u_w}$ is the compression extending $J_{u_w}$ in the lattice group  $G$.
By definition of $u_w$, we have
\[
p_{\lambda(w)}=p_0\oplus \bigoplus_{\tilde w<w} u_{\tilde w},\quad
p_{\lambda(w)}'=\bigoplus_{\tilde w\ge w} u_{\tilde w},
\]
where we sum over strings of length $l(\tilde w)=l(w)$. Since $k(\tilde w)$ and $k(w)$ are
integers such that $k(\tilde w)< k(w)$ if $\tilde w<w$, we have in this case
\[
\tilde U_{u_{\tilde w}}(2^{l(w)}a-k(w)1)\le \tilde U_{u_{\tilde w}}(2^{l(\tilde
w)}a-k(\tilde w)1-1)\le 0.
\]
Since also $\tilde U_{p_0}(2^{l(w)}a-k(w)1)=-k(w)p_0\le 0$ by (ii), we obtain
\[
\tilde U_{p_{\lambda(w)}}(2^{l(w)}a-k(w)1) = \tilde U_{p_0}(2^{l(w)}a-k(w)1)  +\sum_{\tilde
w<w} \tilde U_{u_{\tilde w}}(2^{l(w)}a-k(w)1) \le 0.
\]
Similarly, $\tilde U_{p_{\lambda(w)}'}(2^{l(w)}a-k(w)1)\ge 0$, this proves (iv').
Conversely, if (iv') holds, then
\[
\tilde U_{u_w}(2^{l(w)}a-k(w)1) =\tilde U_{p_{\lambda(w+1)}}\tilde
U_{p_{\lambda(w)}'}(2^{l(w)}a-k(w)1)\ge 0
\]
and also
\[
\tilde U_{u_w}(2^{l(w)}a-k(w)1)-u_w =\tilde U_{p'_{\lambda(w)}}\tilde
U_{p_{\lambda(w+1)}}(2^{l(w)}a-k(w+1)1) \le 0,
\]
so that $\tilde U_{u_w}(2^{l(w)}a-k(w)1)\in C$ for all $w\in \mathcal B$, this is
equivalent to (iv).
The proof is now finished similarly as the proof of Theorem \ref{thm:rsd}.

\end{proof}

\subsection{Further properties of the rational spectral resolution}

We have obtained a well defined binary spectral resolution  of elements in a spectral effect
algebra $E$, which can be extended to a rational spectral resolution if $E$ is
archimedean.\footnote{In fact, archimedeanity of $E$ is not needed, as shown in the
previous version of this work \cite{JenPul_previous}, but the constructions are more
complicated.}
There are some properties of the spectral resolution in $E(\mathcal H)$ that we would
like  to hold also in this general case. Specifically, we would expect that
\begin{enumerate}
\item[(a)] the element $a$ is uniquely determined by its spectral resolution,
\item[(b)] for any projection $q\in P$, we have $a\in C(q)$ iff $p_{a,\lambda}\in C(q)$ for all $\lambda\in \mathbb Q$.
\end{enumerate}
At present, we can show these properties under the assumption of existence of a separating
family of states on $E$.

\begin{prop}\label{prop:spectral_states} Let $E$ be spectral, $a\in E$ and  let $\{c_w\}$ and
$\{u_w\}$ be the families associated with the binary spectral resolution
$\{p_{a,\lambda(w)}\}_{w\in \mathcal B}$ of $a$. Then for
any state  $s$ on $E$:
\begin{enumerate}
\item $s(a)=\lim_{n\to \infty} \sum_{l(w)=n}\lambda(w)s(u_w)$ and $s(a)=0$
if and only if $s(u_w)=0$ for all $w\in \mathcal B$ such that $\lambda(w)\ne 0$.

\item If $E$ is archimedean and $s$ is $\sigma$-additive, then $s(a)=0$ if and only if $s(a\dg)=0$.

\item If $b\in E$ is  such that $p_{a,\lambda(w)}=p_{b,\lambda(w)}$ for all $w\in \mathcal
B$, then
$s(a)=s(b)$.
\end{enumerate}

\end{prop}

\begin{proof}
Let $C$ be any C-block containing $a$ and let $G$ be the universal group. Since $c_w\le
u_w$, $c_w$ with  $w\in \{0,1\}^n$ are summable in $E$ by Lemma \ref{lemma:uw} and using Lemma
\ref{lemma:cw} we obtain in $G$:
\begin{equation}\label{eq:cn}
c_n:=\bigoplus_{l(w)=n} c_w=\sum_{l(w)=n} \tilde U_{u_w}(2^{n}a-k(w)1) =
2^{n}a-\sum_{l(w)=n} k(w)u_w.
\end{equation}
Note that the restriction  $s|_{C}$ is a state on $C$ and hence extends to a state on
$G$, so that
\[
s(c_n)=2^ns(a)-\sum_{l(w)=n}k(w)s(u_w).
\]
The first part of (i) follows from $c_n\in E$, so that $0\le s(c_n)\le 1$. This also
implies that  if $s(a)=0$, then
 we must have $s(u_w)=0$ for all $w$ such that $\lambda(w)\ne 0$. The converse follows
 from the first part of (i).

For (ii), assume $s(a)=0$. Let $0^n$ denote a string of zeros of length $n$, then by (i)
and Lemma \ref{lemma:uw},
\[
0=s(\bigoplus_{l(w)=n, w\ne 0^n} u_w)=s(a\dg\wedge u'_{0^n})=s(a\dg\wedge
p_{2^{-n}}'),
\]
here we used the fact that $u_w=p_{\lambda(w+1)}\wedge p_{\lambda(w)}'$ and
$\lambda(0^{n}+1)=\lambda(0^{n-1}1)=2^{-n}$. Since $E$ is archimedean and $s$ is
$\sigma$-additive, we get
 \[
0=s(a\dg\wedge(\bigvee_{n} p'_{2^{-n}}))=s(a\dg\wedge p'_{0})=s(a\dg).
\]
The converse is obvious from $a\le a\dg$.
The statement (iii) is immediate from (i).

\end{proof}

\begin{prop}\label{prop:Jcommuting}
Let $q\in P$ be a projection commuting with all binary spectral projections of $a$. Then
for any state $s$ on $E$,
\[
s(a)=s(J_q(a)+J_{q'}(a)).
\]
\end{prop}

\begin{proof}
Put  $J:=J_q+J_{q'}$, then $J: E\to E$ is a
unital additive map preserving any element $r$ in the Boolean subalgebra generated by the
binary spectral projections of $a$.
We also have  $J_rJ=JJ_r$ for any such element and  therefore
\[
J(a)=J(J_r(a)+J_{r'}(a))=J_r(J(a))+J_{r'}(J(a)),
\]
so that  $J(a)\in C(r)$. Let $\tilde C$ be a C-block containing $J(a)$ and all $r$ an let $\tilde G$ be its universal
group. Our aim is to show that for $c_n$ as in \eqref{eq:cn}, $J(c_n)\in \tilde C$ for all
$n$ and, as an element of $\tilde G$,
\begin{equation}\label{eq:Jcn}
J(c_n)=2^nJ(a)-\sum_{l(w)=n}k(w)u_{w}.
\end{equation}
Exactly as in the proof of Proposition  \ref{prop:spectral_states}, this then implies that
\[
s(J(a))=\lim_{n\to \infty}\sum_{l(w)=n}\lambda(w)s(u_{w})=s(a).
\]
Note that if $C$ is a C-block containing $a$, then it  is not clear in general whether  $J$ maps $C$ into
$\tilde C$ (or into any ordered unital group). We therefore cannot use any extension of $J$ to the universal group of
$C$.

Let us first show by induction that $J(c_{w})\in \tilde C$ for all $w$. This is clearly true
 for $w=\varepsilon$. Assume that $J(c_w)\in \tilde C$  holds  for some $w$. Then  $J(c_w)\le
 J(u_w)=u_w$ and it is easily seen that, since $u_{w0}\in P_\le(c_w, c_w'\wedge u_w)$, we
  have $u_{w0}\in P_\le(J(c_w),J(c_w)'\wedge u_w)$ and
 \begin{align*}
J((c_w'-c_w)_+\wedge u_w)&=J(J_{u_{w0}}(c_w')\ominus
J_{u_{w0}}(c_w)) =J_{u_{w0}}(J(c_w)'\wedge u_w)\ominus J_{u_{w0}}(J(c_w))\\ &=(J(c_w)'\wedge
u_w-J(c_w))_+=(J(c_w)'-J(c_w))_+\wedge u_w\in \tilde C.
 \end{align*}
Similarly, $J((c_w-c_w'\wedge u_w)_+)=(J(c_w)-J(c_w)')_+\wedge u_w\in \tilde C$ so that $J(c_{w0}), J(c_{w1})\in \tilde C$.

The elements $J(c_{w})$, $l(w)=n$  are summable and we have
\[
J(c_n)=\bigoplus_{l(w)=n} J(c_{w})\in \tilde C.
\]
 We next prove \eqref{eq:Jcn}. Again, we proceed by induction. The statement is clear for
 $l(w)=0$: $J(c_\varepsilon)=J(a)$. Assume that it holds for $l(w)=n$.
 Then
 \begin{align*}
 \sum_{l(w)=n+1}& J(c_w)=\sum_{l(w)=n} (J(c_{w0})+J(c_{w1}))\\
 &=\sum_{l(w)=n}\biggl( u_{w0}- \tilde{\tilde
 U}_{u_w}\bigl((J(c_w)'-J(c_w))_+ + (J(c_w)-J(c_w)')_+\bigr)\biggr)\\
 &=\sum_{l(w)=n}(u_{w0}+\tilde{\tilde U}_{u_w}(J(c_w)-J(c_w)'))=2J(c_n)-\sum_{l(w)=n}u_{w1},
\end{align*}
here $\tilde{\tilde U}_{u_w}$ is the extension of $J_{u_w}|{\tilde C}$ to $\tilde G$. 
Using the induction assumption, we get
\begin{align*}
J(c_{n+1})&=2^{n+1}J(a)-\sum_{l(w)=n}(2k(w)u_w+u_{w1})=2^{n+1}J(a)
   -\sum_{l(w)=n}(2k(w)u_{w0}+(2k(w)+1)u_{w1})\\
&=2^{n+1}J(a)-\sum_{l(w)=n+1}k(w)u_w.
\end{align*}

\end{proof}

\begin{theorem}\label{thm:separating} Assume that $E$ is spectral and has a separating set of states. Then
\begin{enumerate}
\item[(i)] For $q\in P$ and $a\in E$, $a\in C(q)$ if and only if $p_{a,\lambda(w)}\in
C(q)$ for all $w\in \mathcal B$.
\item[(ii)] If $a,b\in E$ are such that $p_{a,\lambda(w)}=p_{b,\lambda(w)}$ for all $w\in
\mathcal B$, then $a=b$.

\end{enumerate}

\end{theorem}

\begin{proof} Immediate from Prop. \ref{prop:spectral_states} and \ref{prop:Jcommuting}.

\end{proof}

\section{Interval effect algebras}

The  notions of compressions, compression bases and their comparability and spectrality
properties were originally introduced in the setting of partially ordered abelian groups.
The necessary definitions and properties  are collected  in Appendix \ref{app:spec}.

Recall that an effect algebra $E$ is an interval effect algebra if it is isomorphic to the
unit interval in a partially ordered unital abelian group (Example \ref{ex:intervalea}).
By \cite{FBinterval}, there is a \emph{universal  group}  $(G,u)$ of $E$ (unique up to
isomorphisms), such that  $E\simeq G[0,u]$ and the following universal property holds: for
any group $H$,  any additive mapping $\Phi: E\to H$ uniquely extends to a group
homomorphism $G\to H$.

It follows from the universal property that any compression $J$ on $E$ extends uniquely to a group
homomorphism $\tilde J$ on $G$ which is again a compression on $G$. Further, any
compression base $(J_p)_{p\in P}$ on $E$ uniquely extends to a compression base $(\tilde
J_p)_{p\in P}$ on $G$. It is clear that also conversely any compression base on $G$
restricts to a compression base on $E$. We obtain a one-to-one correspondence between
compression bases on $E$ and on its universal group. It is a natural question what is the
relation between spectrality on $E$ with a fixed compression base $(J_p)_{p\in P}$ and spectrality of the
universal group with the extension $(\tilde J_p)_{p\in P}$. 

The defining  properties of a spectral group $(G,u)$ with a compression
base are the Rickart property and general comparability property defined in Appendix
\ref{app:spec}. Since the Rickart property of $(G,u)$ is equivalent to the projection
cover property in $E$ if general comparability holds, the basic question to be answered is
whether the b-comparability property in $E$ is equivalent to general comparability for
$G$. We have already answered this question in the affirmative in the case of
interpolation groups resp. effect algebras with RDP in Sec. \ref{sec:rdp}. The next example shows, however, that it is
not true in general. 

\begin{example}\label{ex:horsum_universal}
Let $E$ be the effect algebra as in Example \ref{ex:horsum} (the horizontal sum of a
Hilbert space effect algebra with itself). We have introduced a compression base in $E$
and noticed in Example \ref{ex:horsum_spectral} that $E$ is spectral. As the horizontal
sum of two interval effect algebras, $E$ is an interval effect algebra as well,
\cite{fgb}. Note however that its universal group cannot be spectral: let $e=(1/2,0)$ and
$f=(0,1/2)$, then $e,f\in E$, $e\ne f$ but $2e=2f=u$. It follows that the universal group
is not torsion free, hence it cannot be spectral by \cite[Lemma 4.8]{Fcomgroup}. 

\end{example}

In the next sections we show that the above equivalence holds for archimedean divisible
effect algebras.

\subsection{Divisible effect algebras}
\begin{definition}
An element $a$ in an effect algebra $E$ is \emph{divisible} if for any $n\in \mathbb N$, there is a unique $x\in E$ such
that $nx=a$, this element will be denoted by $a/n$. If every element in $E$ is divisible, we say that $E$ is
\emph{divisible}.

\end{definition}

By \cite{Pudiv1}, any divisible effect algebra is an interval effect algebra whose
universal group is divisible. In
particular, for any $\lambda\in [0,1]\cap \mathbb Q$ and $a\in E$, there is an element $\lambda a\in E$ and the
 map $(\lambda,a)\mapsto \lambda a$ has the properties making it an \emph{effect
 module} over $[0,1]\cap \mathbb Q$, see \cite{jacobs} for the definition.  Conversely,
 the unit interval in a divisible group $(G,u)$ is a divisible effect
 algebra.

\begin{theorem}\label{thm:comparability_divisible} If $E$ has  the b-comparability
property, then $G$ has general comparability. If $E$ is archimedean, the converse is also true.

\end{theorem}

\begin{proof} For any $g\in G$, there are some  $K,N\in \mathbb N$ such that $K\le N$ and  $a:=(1/N)(g+Ku)\in E$. Clearly, any projection commutes with $a$ if and only if it commutes
with $g$, this implies that $P(a)=P(g)$.
By the b-comparability property, there is some $p\in P_\le(a, (K/N)u)$, which means that $p\in P(a)=P(g)$ and
$J_p(a)\ge J_p((K/N)u)$, $J_{p'}(a)\le J_{p'}((K/N)u)$, which clearly implies that $p\in P_\pm(g)$.

Conversely, assume that $E$ is archimedean and  that $G$ has general comparability. Then
by Lemma \ref{lemma:app_arch}, the order unit seminorm $\|\cdot\|$ is a norm in $G$.  By
Lemma \ref{lemma:app_sum}, we obtain that for every $g\in
G$ and $n\in \mathbb N$, there is a
rational linear combination $\sum_i \xi_iu_i$ of elements in $P(g)$ such that $\|g-\sum_i\xi_iu_i\|\le 1/n$, so that
 $g$ is in the norm closure of the $\mathbb Q$-linear span of $P(g)$. Let $a\in E$ and $q\in P$ be such
that $P(a)\subseteq C(q)$. Then also any $\mathbb Q$-linear combination of elements in $P(a)$ is in $C_G(q)$ and since
$C_G(q)$ is closed in $\|\cdot\|$ by \cite[Cor.~3.2]{Forc}, it follows that $a\in C(q)$.
This shows that $E$ has the b-property. Let $a,b\in E$, $aCb$ and let $g=a-b$,
 $p\in P_\pm(g)$. Clearly, $J_p(a)\ge J_p(b)$ and $J_{p'}(a)\le J_{p'}(b)$. Further, it is easily seen that $g\in
C_G(P(a))$, so that $P(a)\subseteq C(p)$ and hence $a\in C(p)$, similarly $b\in C(p)$. Assume next that
 $q\in P$ is such that
 $a,b\in C(q)$, then $g\in C_G(q)$ so that $p\in C(q)$, it follows that $p\in CP(a,b)$. Putting all together, we obtain
 $p\in P_\le (a,b)\ne \emptyset$. Hence $E$ has the b-comparability property.

\end{proof}

 If  $E$ is divisible then its
universal group $G$ is archimedean (see  Appendix \ref{app:spec}) if and only if for $a,b,c\in E$, $a\le b\oplus \frac 1n
c$ for all $n\in \mathbb N$ implies that $a\le b$. In general this property is stronger
than archimedeanity of $E$, see \cite{GPBB}, so let us call a divisible effect algebra
with this property   \emph{strongly archimedean}. \label{key:sa}  As a
consequence of the above result, we show that for divisible effect algebras with
b-comparability these two notions coincide.

\begin{corollary}\label{coro:Arch} An archimedean divisible effect algebra with b-comparability is strongly
archimedean.

\end{corollary}

\begin{proof} Since b-comparability in $E$ implies general comparability for the universal
group, the results follows from Lemma \ref{lemma:app_arch}.

\end{proof}

\begin{theorem}\label{thm:divisible_spectral} Let $E$ be an archimedean divisible effect
algebra and let $(G,u)$ be its universal group. Then $E$ is spectral if and only if $G$ is spectral.
 In this case, for $a\in E$:
 \begin{enumerate}
\item the rational spectral resolution obtained in $E$ and in $G$ is the same,
\item $a$ is the norm limit
\[
a=\lim_{n\to \infty} \sum_{l(w)=n} \lambda(w) u_{w},
\]
where $u_{w}=p_{a,\lambda(w+1)}\wedge
p_{a,\lambda(w)}'$,  $w\in \mathcal B$.
\item For $b\in E$, $a=b$ if and only if $p_{a,\lambda(w)}=p_{b,\lambda(w)}$, $w\in \mathcal B$.
\item For $a\in E$, $q\in P$, $a\in C(q)$ if and only if $p_{a,\lambda(w)}\in C(q)$ for
all $w\in \mathcal B$.
\end{enumerate}

\end{theorem}

\begin{proof} The first statement is clear from  Theorem \ref{thm:comparability_divisible} and the fact that the projection cover
property of $E$ and $G$ are clearly equivalent.

Assume that $E$ is spectral.
For (i), let $C$ be a C-block containing $a$ and let $\tilde G$ be its universal group, which is a subgroup in $G$ and
a spectral lattice ordered group in its own right. By Theorem
\ref{thm:spectral_archimedean}, the rational spectral resolution of $a$ in $E$ coincides
with the spectral resolution in $\tilde G$, which in turn is the same in $G$, by the
characterization in Theorem \ref{thm:rsd}.

To proceed, notice that $G$ is archimedean by Corollary \ref{coro:Arch}, so that
$\|\cdot\|$ is a norm and $E$ has a separating set of states, by Lemma
\ref{lemma:app_arch}.  The statement  (ii) now follows from  \eqref{eq:cn}, (iii) and (iv)
are
from Theorem \ref{thm:separating}.

\end{proof}

It is easily observed from Theorem \ref{thm:rsd} that for divisible effect algebras, the
rational spectral resolution has a characterization resembling the  spectral resolutions in order unit spaces in spectral
duality, see \cite[Thm. 8.64]{AlSh}, in particular, in von Neumann algebras.

\begin{corollary} Let $E$ be a  divisible archimedean spectral effect algebra. Then the rational spectral resolution of an
element $a\in E$ is the unique parametrized family $\{p_\lambda\}_{\lambda\in \mathbb Q}$ in $P$  such that
\begin{enumerate}
\item $p_\lambda=0$ for $\lambda< 0$ and $p_\lambda=1$ for $\lambda\ge1$,
\item if $\lambda\le \mu$ then $p_\lambda\le p_\mu$,
\item $\bigwedge_{\lambda<\mu}p_\mu=p_\lambda$,
\item $p_\lambda\in PC(a)$,
\item $J_{p_\lambda}(a)\le \lambda p_\lambda$, $J_{p'_\lambda}(a)\ge \lambda p'_\lambda$.

\end{enumerate}
\end{corollary}

\subsection{Convex effect algebras}

An effect algebra  $E$ is \emph{convex}  \cite{GuPu} if for every $a\in E$ and $\lambda \in [0,1] \subset {\mathbb R}$ there is an element $\lambda a\in E$ such that for all $a,b\in E$ and all $\lambda, \mu \in [0,1]$ we have

\begin{enumerate}
\item[(C1)] $\mu(\lambda a)=(\lambda \mu)a$.
\item[(C2)] If $\lambda + \mu\leq 1$ then $\lambda a\oplus \mu a\in E$ and $(\lambda +\mu)a=\lambda a\oplus \mu a$.
\item[(C3)] If $a\oplus b\in E$ then $\lambda a\oplus \lambda b \in E$ and $\lambda(a\oplus b)=\lambda a\oplus \lambda b$.
\item[(C4)] $1a=a$.
\end{enumerate}

A convex effect algebra is convex in the usual sense: for any $a,b\in E$, $\lambda \in [0,1]$, the element $\lambda a\oplus
(1-\lambda)b\in E$. An important example of a convex effect algebra is the algebra $E(\mathcal H)$ of Hilbert space
effects, Example \ref{ex:effects}. Clearly, any convex effect algebra is divisible.

Let $V$ be  an ordered real linear space with positive cone $V^+$. Let $u\in V^+$ and let us form the interval effect algebra $V[0,u]$.
 A straightforward
verification shows that  $(\lambda,x)\mapsto \lambda x$ is a convex structure on $V[0,u]$, so $V[0,u]$ is a
convex effect algebra which we call a \emph{linear effect algebra}. By \cite[Thm. 
3.4]{GPBB}, the universal group of a convex effect algebra is the additive group of an
ordered  real linear space $(V,V^+)$ with an order unit $u$, so that any convex effect
algebra is isomorphic to the linear effect algebra $V[0,u]$. Moreover, this isomorphism is
\emph{affine}, which means that it preserves the convex structures.

The next theorem describes the relations between order unit spaces, ordering sets of states and strongly archimedean convex effect algebras (cf. the definition on p. \pageref{key:sa}).

\begin{theorem}\label{th:archim}{\rm \cite[Thm. 3.6]{GPBB}} Let $E\simeq V[0,u]$ be  a
convex effect algebra with universal group given by the ordered linear space $(V,V^+)$
with order unit $u$.  Then the following statements are equivalent. {\rm(a)} $E$ possesses an ordering set of states. {\rm(b)} $E$ is strongly archimedean. {\rm(c)} $(V, V^+,u)$ is an order unit space.
\end{theorem}

Compressions, compression bases and spectrality in order unit spaces were studied in
\cite{FPspectres,JenPul}. In addition to the properties of compressions in
partially ordered unital abelian groups in Appendix \ref{app:spec}, the maps are also
required to be linear, which is equivalent to the condition that the restrictions to the
unit interval are affine. The next result shows that if $E$ is strongly archimedean, this
property always holds.

\begin{theorem}\label{th:archaf} Let $\phi:E \to F$ be an additive mapping from a convex  effect algebra $E$ into a strongly archimedean convex effect algebra $F$. Then $\phi$ is affine.
\end{theorem}

\begin{proof}  As $\phi$ is additive, it is easy to show that $\alpha \phi(a)=\phi(\alpha a)$ for any rational $\alpha$. Assume that $\alpha\leq \beta\leq \gamma$, where $\alpha$ and $\gamma$ are rational numbers and  $\beta \in [0,1]\subset {\mathbb R}$.
Then

\[
\alpha \phi(a)=\phi(\alpha a)\leq \phi(\beta a)\leq \phi(\gamma a)= \gamma \phi(a).
\]
From this we get
\[
\phi(\beta a)\leq \beta \phi(a)\oplus (\gamma - \beta)\phi(a).
\]
For every $n\in {\mathbb N}$ we find $\gamma \in {\mathbb Q}$ such that $\gamma-\beta \leq
\frac{1}{n}$. Since $E$ is strongly archimedean, this entails that
\[
\phi(\beta a)\leq \beta \phi(a).
\]

Moreover,
\[
\beta \phi(a)\ominus \phi(\beta a) \leq \gamma \phi(a)\ominus \alpha \phi(a)=(\gamma -\alpha)\phi(a),
\]
$\gamma -\alpha =(\gamma -\beta)+(\beta -\alpha)$, and we find $\gamma$ and $\alpha$ such that $\gamma-\alpha \leq \frac{1}{n}$ for every $n$. This yields
\[
\phi(\beta a)=\beta \phi(a).
\]
\end{proof}

\begin{corollary}\label{co:archaf} Every retraction on a convex  strongly archimedean effect algebra $E$ is affine.
\end{corollary}

From the above results and the fact that convex effect algebras are divisible,  we now easily derive:

\begin{theorem}\label{thm:convex_spectral} Let $E$ be a convex and strongly archimedean effect algebra and let $(V,V^+,u)$ be the corresponding order unit space. Then
\begin{enumerate}
\item Compressions (compression bases) on $E$ uniquely correspond to compressions (compression bases) on $(V,V^+,u)$
(in the sense of \cite{FPspectres}).
\item $E$ is spectral if and only if $(V,V^+,u)$ is spectral (in the sense of \cite{FPspectres}).

\end{enumerate}

\end{theorem}

The spectral resolutions in order unit spaces are indexed by all real numbers: the spectral projections are defined by
\[
p_{a,\lambda} = (a-\lambda 1)_+^*,\qquad \lambda\in \mathbb R.
\]
Again by divisibility and the properties of the spectral resolutions in order unit spaces
\cite[Thm. 3.5 and Remark
3.1]{FPspectres}, we obtain using Corollary \ref{coro:Arch}:

\begin{theorem} Let $E$ be a convex archimedean effect algebra with universal group given
by the ordered linear space $(V,V^+)$ with order unit $u$.  Assume that $E$ is spectral  and let $\{p_{a,\lambda}
\}_{\lambda\in \mathbb Q}$ be the rational spectral resolution of an element $a\in E$. For $\lambda \in \mathbb R$, put
\[
p_{a,\lambda}:=\bigwedge_{\mu>\lambda, \ \mu\in \mathbb Q} p_{a,\mu}.
\]
Then
\begin{enumerate}
\item $E$ is strongly archimedean, so that $(V,V^+,u)$ is an order unit space.
\item $(V,V^+,u)$ with the extended compression base is a spectral order unit space.
\item $\{p_{a,\lambda}\}_{\lambda\in \mathbb R}$ coincides with the spectral resolution of $a$ in $(V,V^+,u)$.
\item $a$ is given by the Riemann Stieltjes type integral
\[
a=\int_0^1 \lambda dp_{a,\lambda}.
\]

\end{enumerate}

\end{theorem}

\begin{example}\label{ex:hilb_convex}
A prominent example of a convex effect algebra is the Hilbert space effect algebra $E(\mathcal H)$. We have already
checked  in
Example \ref{ex:hilb_spec} that $E(\mathcal H)$ is spectral. The corresponding spectral resolution coincides with the
usual spectral resolution of Hilbert space effects.

\end{example}

\begin{example}\label{ex:convexMV} A convex archimedean MV-effect algebra $E$ is isomorphic to a dense subalgebra in the algebra $C(X,[0,1])$ of continuous functions
$X\to [0,1]$ for a compact Hausdorff space $X$, \cite[Thm. 7.3.4]{DvPu}. If $E$ is norm-complete (in the supremum norm) then $E\simeq C(X,[0,1])$.  By the above results, $E$ is spectral if and only if the space $C(X,\mathbb R)$ of all continuous real functions, with its natural order unit space structure, is spectral in the sense
of \cite{FPspectres}. It was proved in \cite{FPmonot} that this is true if and only if $X$ is basically disconnected,
which is equivalent to the fact that $E$ is monotone $\sigma$-complete, see \cite{Good}.

\end{example}

\begin{example}\label{ex:cs} Let $(X,\|\cdot\|)$ be a (real) Banach space and let
\[
E=\{ (\lambda, x):\ \lambda\in \mathbb R, x\in X,\ \|x\|\le \lambda\le 1-\|x\|\}.
\]
Then $E$ is a convex archimedean effect algebra. The corresponding order
unit space was considered already  in \cite{berd} and was subsequently called a \emph{generalized spin factor}
in \cite{beod}.

By Theorem \ref{thm:convex_spectral} and \cite[Thm. 6.5]{JenPul}, $E$ is
spectral if and only if $X$ is reflexive and strictly convex. Moreover, it follows from  \cite[Thm.1]{beod}, see also
\cite[Thm. 6.6]{JenPul}, that in addition $X$ is also smooth if and only if $E$ is spectral
in the stronger sense derived from spectral duality due to Alfsen-Schultz \cite{AlSh}.

\end{example}

%We finish by a simple example of a divisible spectral effect algebra that is not convex.
%
%\begin{example}\label{ex:dvi_not-conv} Let $X$ be a totally disconnected compact Hausdorff
%space and let $E$ be the set of all finite valued continuous functions $X\to [0,1]\cap
%\mathbb Q$, with  $f\oplus g=f+g$ if $f +g\le 1$ and undefined otherwise. Then $E$ is a
%divisible archimedean effect algebra that is not convex.  It is easily seen that the center
%of $E$ is precisely  the set of characteristic functions of clopen 
%subsets in $X$ and that $E$ with its central compression base is spectral, cf. \cite[Ex.
%1.7]{FPspectres}.
%
%\end{example}
%

\section{Conclusions and future work}

We have shown that in a spectral effect algebra in the sense of \cite{Pucompr}, any element has a well defined binary spectral resolution in terms of specific principal elements called projections. In the case that $E$ is archimedean, we found a characterization of the spectral resolution as the unique family of projections satisfying certain conditions.
There are important properties expected from a spectral resolution: that it uniquely determines the corresponding element as well as all projections compatible with it.
We were able to prove these properties only under the condition that a separating family of states exists. It is a question whether this seemingly rather strong condition can be dropped, but note that even  for the
proof of corresponding properties for groups (see \ref{thm:rsd2}, \cite[Thm 4.4]{Forc}) it
was assumed that $G$ is archimedean, which in this case is equivalent to existence of an
\emph{ordering} set of states.

We also looked at some special cases of interval effect algebras: (effect algebras with
RDP, divisible
and convex effect algebras) and proved that spectrality and spectral resolutions are
equivalent to the respective notions in their universal group. However, we found an
example where this is not true. We conjecture that such equivalence holds for unit intervals
in \emph{unperforated} partially ordered unital abelian groups. Note that under such an equivalence, for the spectral resolution to be right continuous, the effect algebra $E$ must have an ordering set of states, this follows from Proposition \ref{prop:right_cont} and Theorem \ref{th:archim}.

In general, it is still not clear what intrinsic properties of $E$ are necessary and
sufficient for existence of a compression base making it spectral. For this, it is
important to characterize spectrality  for MV-effect algebras, which are the basic building blocks for spectral effect algebras. 
In this case, some sufficient conditions are known, such as $E$ being  monotone  $\sigma$-complete or a 
 boolean algebra, but a general characterization (even under some additional assumptions
 such as archimedeanity or norm completeness) is not known. The question how the different
 C-blocks are connected and intertwined by the compressions in the compression base seems much more challenging.

For convex spectral effect algebras one can consider the possibility of introducing the functional calculus for its elements. A further  possible line of research is
to look at spectrality in some special cases  of (convex) effect algebras. For example, in
\cite{wetering2} a spectral theorem was proved for elements in a  normal  sequential effect algebra, 
one can show that this corresponds to our results. A related question is the characterization of convex sequential effect algebras, in particular the unit interval in a JB-algebra, in the context of spectrality.

\appendix

\section{Spectrality  in ordered groups}\label{app:spec}
 \setcounter{equation}{0}
\renewcommand{\theequation}{\thesection.\arabic{equation}}

The notions considered in this paper for an effect algebra $E$ were inspired by
similar  notions   for an unital ordered abelian group $(G,u)$,  defined and studied by Foulis in
\cite{Funig, Forc}. Here we collect some basic results on spectrality in this setting.

If $E=G[0,u]$ is the unit interval in $(G,u)$,  we can say, in short, that
a compression on $G$ is an order-preserving group endomorphism  $J: G\to G$ such that the restriction of $J$ to $E$ is a
compression on $E$. Similarly, a compression base  $(J_p)_{p\in P}$ in $G$  is a family of compressions in $G$ such that the restrictions
$(J_p|_E)_{p\in P}$ form a compression base in $E$.

We say that $G$ is a \emph{compressible group} if every retraction is a compression and each compression
is  uniquely determined by its focus. In this case, the set of all compressions is a compression base, see
\cite{Funig} for the definition, proofs and further results.

Note that in \cite{Forc} spectrality was introduced in the setting of compressible groups. We will use a slightly more
general assumption that we have a distinguished  compression base $(J_p)_{p\in P}$ in
$G$, that will be fixed throughout the present section. The proofs of the statements below remain the same.

The definitions of compatibility and commutants for $G$ are
analogous to those in Sec. \ref{sec:commutants} for effect algebras. For $p\in P$, we will use the notation
$C_G(p)$ for the set of all elements $g\in G$ compatible with $p$. The definitions of $PC(g)$ and the bicommutant $P(g)$ extend straightforwardly to all $g\in G$.

\medskip
\noindent
\textbf{General comparability.}
 We say that $G$ has \emph{general comparability} if for any $g\in G$, the set
\[
P_\pm(g):=\{p\in P(g):\ J_{p'}(g)\le 0 \le J_p(g)\}
\]
is nonempty. In this case, $G$ is   unperforated  (that is, $ng\in G^+$ implies $g\in G^+$ for $g\in G$ and $n\in
\mathbb N$) \cite[Lemma 4.8]{Fcomgroup}.

Since $u$ is an order unit in $G$, it defines an order unit seminorm in $G$ as
\[
\|g\|:=\inf\{\frac nk:\ k,n\in \mathbb N, -nu\le kg\le nu\}.
\]
We say that $G$ is \emph{archimedean} \cite{Good} if
 $g,h \in G$ and $g\le nh$ for all $n\in \mathbb N$ implies $g\le 0$, we remind the reader
 that this is in general stronger than archimedeanity of $E$.

The following result was basically proved in \cite{Forc}.

\begin{lemma}\label{lemma:app_arch} Assume that $G$ has general comparability. Then the following are equivalent.
\begin{enumerate}
\item $E$ is archimedean.
\item $G$ is archimedean.
\item $\|\cdot\|$ is a norm on $G$.
\item $G$ has an order determining set of states.
\item $G$ has a separating set of states.

\end{enumerate}

\end{lemma}

\begin{proof} Assume (i) and let $g,h\in G$ be such that $ng\le h$ for all $n\in \mathbb N$. Let $p\in P_\pm(g)$, then
 $J_p(g)\in G^+$ and hence $0\le nJ_p(g)\le J_p(h)$ for all $n\in \mathbb N$. Since $u$ is an order unit, there is some
$m\in \mathbb N$ such that $J_p(h)\le mu$, and then $nJ_p(g)\le mu$ for all $n$. In particular,
$mkJ_p(g)\le mu$ for all $k\in \mathbb N$ and since $G$ is unperforated, we get $kJ_p(g)\le u$, for all $k\in \mathbb
N$. By (i), this implies that $J_p(g)=0$, hence
\[
g=J_{p'}(g)\le 0,
\]
so that (ii) holds. The equivalence of (ii)-(iv) was proved in \cite[Thm. 4.14, Thm. 7.12]{Good}, see also \cite[Thm. 3.3]{Forc}. The
implications (iv) $\implies$ (v) and (v) $\implies$ (i) are easy.

\end{proof}

\begin{lemma}\label{lemma:app_sum} Assume that $G$ has general comparability and let $g\in G$. Let $n\in
\mathbb N$, $\epsilon >0$  and $m_0,\dots, m_N\in \mathbb Z$ be such that  $m_i\le m_{i+1}$ and $m_0\le -\|g\|$, $m_N\ge \|g\|$.
Then there are elements $u_1,\dots,u_N\in P(g)$, $\sum_i u_i=u$ such that
\[
\|ng-\sum_i m_i u_i\|\le \max_i (m_{i+1}-m_i).
\]

\end{lemma}

\begin{proof} This proof is very similar to the proof of \cite[Thm. 3.22]{JenPul}, we give the proof here for
completeness. Similarly as in \cite[Lemma 3.21]{JenPul}, we may construct a nonincreasing sequence $q_i$, $i=0,\dots N$ such that
$q_i\in P_\pm(ng-m_{i}u)$. Put $q_0=u$, $q_N=0$  and for $i=0,\dots,N-2$ put
\[
q_{i+1}:= r_{i+1}\wedge q_i
\]
where $r_{i+1}$ is any element in $P_\pm(ng-m_{i+1}u)$. We will check that $q_{i+1}\in P_\pm(ng-m_{i+1}u)$ (cf \cite[Thm.
3.7]{Fgc}). Indeed, we have
\[
J_{q_{i+1}}(ng-m_{i+1}u)=J_{q_{i+1}}J_{r_{i+1}}(ng-m_{i+1}u)\ge 0
\]
and since $q_{i+1}'=r_{i+1}'+r_{i+1}\wedge q_i'$ and all these elements are in $P(g)$, we obtain
\[
J_{q_{i+1}'}(ng-m_{i+1}u)=(J_{r_{i+1}'}+J_{r_{i+1}}J_{q_i'})(ng-m_{i+1}u) \le J_{r_{i+1}'}(ng-m_{i+1}u)+J_{r_{i+1}}J_{q_i'}(ng-m_{i}u)\le
0.
\]
Next, for $i=1,\dots,N$, put
\[
u_i:=q_{i-1}-q_i=q_{i-1}\wedge q_i',
\]
then $u_i\in P(g)$ and $\sum_{i=1}^N u_i= u$. We also have
\begin{align*}
J_{u_i}(ng-m_iu)=J_{q_{i-1}}J_{q_i'}(ng-m_iu)\le 0, \quad 
J_{u_i}(ng-m_{i-1}u)= J_{q_i'} J_{q_{i-1}}(ng-m_{i-1}u)\ge 0,
\end{align*}
so that $m_{i-1}u_i\le J_{u_i}(ng)\le m_i u_i$ and hence
\[
 -(m_i-m_{i-1})u_i\le J_{u_i}(ng)-m_iu_i\le 0.
\]
Summing over $i$ now gives the result.

\end{proof}

\medskip
\noindent
\textbf{Rickart property.}
We say that $G$ has the \emph{Rickart property} (or that $G$ is \emph{Rickart}) if
there is a mapping $*: G\to P$, called the \emph{Rickart mapping}, such that for all $g\in G$ and $p\in P$, $p\leq
g^*\,\Leftrightarrow\, g\in C(p), J_p(g)=0$. In this case, the unit interval $E$ has the projection cover property, with the projection cover obtained as
\begin{equation}\label{eq:projcover_rickart}
a\dg=a^{**}=(a^*)',\qquad \mbox{for } a\in E.
\end{equation}
In particular, $P$ is an OML. If $G$ has general comparability, then the projection cover property of $E$ is equivalent to the Rickart property of $G$, see \cite[Thm. 6.5]{Fgc}.

Some important properties of the Rickart mapping are collected in the following lemma.
\begin{lemma}\label{le:rickmap} {\rm \cite[Lemma 6.2]{Fgc}} For all $g,h\in G$ and all $p\in P$ we have:
{\rm(i)} $g^*\in P(g)$ and
$J_{g^*}(g)=0$. {\rm(ii)} $p^*=p'$. {\rm(iii)} $g^{**}:=(g^*)^*=(g^*)'$.   {\rm(iv)}
$0\leq g\leq h \implies g^{**}\leq h^{**}$. {\rm(v)} $J_p(g)=g \,\Leftrightarrow \, g^{**}\leq p$.
\end{lemma}

\medskip
\noindent
\textbf{Orthogonal decompositions.} An \emph{orthogonal decomposition} of an element $g\in
G$ is defined as
\[
g=g_+-g_-,\qquad g_+,g_-\in G^+
\]
such that there is a projection $p\in P$ satisfying $J_p(g)=g_+$ and $J_{p'}(g)=-g_-$. A projection $p\in P$ defines an
orthogonal decomposition of $g$  if and only if $g\in C_G(p)$ and $J_{p'}(g)\le 0\le J_p(g)$.

\begin{lemma}\label{lemma:orthog_decomp} Assume that $G$ has general comparability. Then
\begin{enumerate}
 \item[(i)]\cite[Lemma 4.2]{Fgc}  For all $g\in G$, there is a unique orthogonal decomposition. This decomposition is
 defined by any element in $P_\pm(g)$.
\item[(ii)]\cite[Thm. 3.1]{Forc} If $G$ has also the Rickart property, then $g_+^{**}\in P_\pm(g)$. Moreover,
$g_+^{**}$  is the smallest projection defining the orthogonal decomposition of $g$.
\item[(iii)] Let $q$ be any projection such that $g\in C_G(q)$. Then
$(J_q(q))_\pm=J_q(g_\pm)$.
\end{enumerate}

\end{lemma}

\begin{proof} The last part of statement (ii) was proved in \cite[Thm. 3.1]{Forc} in a slightly weaker form, namely that
$g_+^{**}$ is the smallest element in $P_\pm(g)$. The difference is that not all projections defining the orthogonal
decomposition are in the bicommutant $P(g)$. However, the proof remains the same: it follows immediately from
$J_p(g_+)=g_+$, which implies $g_+^{**}\le p$, see Lemma \ref{le:rickmap} (v).
For (iii), let $p\in P_\pm(g)$, then $p$ commutes with $q$ and $J_p(g)=g_+$,
$J_{p'}(g)=-g_-$. We then have
\begin{align*}
J_p(J_q(g))=J_q(J_p(g))=J_q(g_+)\ge 0,\qquad  J_{p'}(J_q(g))=J_q(J_{p'}(g))=-J_q(g_-)\le 0.
\end{align*}
It follows that $p$ defines the orthogonal decomposition of $J_q(g)$ and
\begin{align*}
J_q(g)_+=J_pJ_q(g)=J_q(g_+),\qquad J_q(g)_-=-J_{p'}(J_q(g))=J_q(-J_{p'}(g))=J_q(g_-).
\end{align*}

\end{proof}

\medskip
\noindent\textbf{Spectrality.}
Assume that $G$ has both the general comparability and the Rickart property. In that case, we will say that $G$ is
\emph{spectral}.
For  $g\in G$ and $\lambda\in \mathbb Q$, let
\begin{equation}\label{eq:app_spectprojs}
p_{g,\lambda}:= ((ng-mu)_+)^*,\qquad \lambda=\frac{m}{n},\ n>0.
\end{equation}
The element  $p_{g,\lambda}$ is well defined, in the sense that it does not depend on the expression
$\lambda=\frac{m}{n}$. The family of projections $(p_{g,\lambda})_{\lambda\in \mathbb Q}$, introduced in
\cite[Def. 4.1]{Forc}, is called the  \emph{rational spectral resolution}  of $g$.
Put
\begin{align*}
l_g&:=\sup\{m/n:\ m,n\in \mathbb Z, 0<n, mu\le ng\}=\sup\{\lambda\in \mathbb Q:\ p_{g,\lambda}=0\}\\
u_g&:=\inf\{m/n:\ m,n\in \mathbb Z, 0<n, mu\ge ng\}=\inf\{\lambda\in \mathbb Q:\ p_{g,\lambda}=u\}.
\end{align*}

The rational spectral resolution has a number of properties resembling the spectral
resolution for self-adjoint operators, see \cite[Thms. 4.1, 4.4]{Forc}.
If $G$ is archimedean, we have the following characterization.

\begin{theorem}\label{thm:rsd} Let $G$ be archimedean and spectral. For $g\in G$, the rational spectral
resolution is the unique family of projections $\{p_\lambda\}_{\lambda\in \mathbb Q}$ such
that
\begin{enumerate}
\item for $\lambda<l_g$, $p_\lambda=0$ and for $\lambda\ge u_g$, $p_{g,\lambda} =u$,
\item for $\lambda<\mu$ we have $p_{g,\lambda}\le p_{g,\mu}$.
\item $\bigwedge_{\lambda<\mu} p_{g,\mu}=p_{g,\lambda}$, for all $\lambda\in \mathbb Q$,
\item for $\lambda=\frac mn$,  $n\in \mathbb N$, $m\in \mathbb Z$, we have
$nJ_{p_{g,\lambda}}(g)\le mp_{g,\lambda}$,  $mp_{g,\lambda}'\le nJ_{p_{g,\lambda}'}(g)$.
\end{enumerate}

\end{theorem}

\begin{proof} Assume that $\{p_\lambda\}$ satisfies the properties (i)-(iv). Let
$\lambda=\frac mn$, $m\in \mathbb Z$, then we have from (iv) that $J_{p_\lambda}(ng-mu)\le
0$ and $J_{p_\lambda'}(ng-pu)\ge 0$, so that $p_\lambda'$ defines the orthogonal
decomposition of $ng-mu$, $J_{p_\lambda'}(ng-pu)=(ng-pu)_+$. Let $\mu>\lambda$, $\mu\in
\mathbb Q$, then there are some $k\in \mathbb N$ and $l_1,l_2\in \mathbb Z$, $l_1<l_2$
such that $\lambda=\frac{l_1}k$ and $\mu=\frac{l_2}k$. By (ii), $p_\mu'\le p_\lambda'$ and
therefore
\begin{align*}
J_{p_\mu'}((kg-l_1u)_+)&=J_{p_\mu'}J_{p_\lambda'}(kg-l_1u)=J_{p_\mu'}(kg-l_2u+(l_2-l_1)u)
=(kg-l_2u)_++(l_2-l_1)p_\mu'\\ &\ge (l_2-l_1)p_\mu'.
\end{align*}

Let $r=(kg-l_1u)_+\dg=(kg-l_1u)_+^{**}\in P_\pm(kg-l_1u)$, then since $r\in P(g)$, $r$
commutes with $p_\mu$ and we obtain
$(l_2-l_1)p_\mu'\le J_{p_\mu'}((kg-l_1u)_+)\le r$. Since $r$ is principal, this entails
$p_\mu'\le r$ for all $\mu>\lambda$. It follows that $r'\le \bigwedge_{\mu>\lambda}
p_\mu=p_\lambda$. On the other hand, by Lemma \ref{lemma:orthog_decomp} (ii), we obtain
$r\le P_\lambda'$, so that $p_\lambda=r'=(kg-l_1u)_+^*$ is the spectral projection of $g$.
The converse is clear from the properties of the spectral resolution, \cite[Thm. 4.1,
4.4]{Forc}.

\end{proof}

\begin{theorem}\label{thm:rsd2} \cite[Thm. 4.4]{Forc}
Let $G$ be archimedean and spectral, $g\in G$.
Then
\begin{enumerate}

\item Each element $g\in G$ is uniquely determined by its rational spectral resolution.
\item For $q\in P$, we have $g\in C(q)$ if and only if $p_{g,\lambda}\in C(q)$ for all $\lambda\in \mathbb Q$.

\end{enumerate}

\end{theorem}

\section*{Acknowledgements}

The authors are indebted to Gejza Jen\v ca for discussions and valuable comments, and to Ryszard Kostecki for pointing
out the references \cite{berd} and \cite{beod}. We are grateful to the anonymous referee
for a number of inspiring remarks that helped us to improve the paper.
The research was supported by the grants VEGA 2/0142/20 and the Slovak Research and Development Agency grant APVV-20-0069.

\end{document}